\documentclass[a4paper, 11pt]{amsart}

\usepackage{longtable, stackrel, multicol, fancybox, tcolorbox}
\usepackage{bm,amsfonts,amsmath,amssymb,mathrsfs,bbm,amsthm} 
\usepackage{graphicx, color,hyperref,eurosym, eucal,palatino}

\definecolor{darkblue}{rgb}{0.15,0.15,0.1}
\definecolor{darkblue2}{rgb}{0.2,0.2,0.9}
\definecolor{superdarkblue}{rgb}{0.2,0.2,0.3}
\hypersetup{colorlinks=true, linkcolor=darkblue,citecolor=darkblue, 
	urlcolor = superdarkblue}
\definecolor{citegreen}{rgb}{0.2,0.2,0.6}
\usepackage{marginnote}
\usepackage{geometry}

\newcommand\bfo{{\bf 0}}

\hypersetup{colorlinks=true, linkcolor=blue,citecolor=citegreen}
\definecolor{citegreen}{rgb}{0.2,0.2,0.6}
\definecolor{darkblue}{rgb}{0.1,0.1,0.6}
\definecolor{md}{rgb}{0.7,0.3,0.}

\usepackage{xcolor}

\newcommand\Omg{\Omega}

\hypersetup{colorlinks=true, linkcolor=blue,citecolor=citegreen}
\definecolor{citegreen}{rgb}{0.2,0.2,0.6}
\definecolor{md}{rgb}{0.7,0.3,0.}

\newcommand\tm{\times}
\newcommand\Dl{\Delta}
\newcommand\OpD{\sfH_{\rm D}^\Omg}
	
\renewcommand\gg{\gamma}

\newcommand\supp{{\rm supp}\,}

\newcommand\RzD{\sfR_{\rm D}^\Omg(z)}
\newcommand\bb{\beta}

\newcommand\bx{\mathbf{x}}

\renewcommand\and{\qquad\text{and}\qquad}
\newcommand\gzy{g_{z,y}^\Omg}
\newcommand\hzy{h_{z,y}^\Omg}
\newcommand\gy{g_y^\Omg}
\newcommand\ggy{\gg_y^\Omg}
\newcommand\hy{h_y^\Omg}

\newcommand\sm{\setminus}
\newcommand\Op{\sfH_{\aa,y}^\Omg}

\newcommand\frm{\mathfrak{h}_{\aa,y}^\Omg}

\newcommand\frmB{\mathfrak{h}_{\aa,\bfo}^\cB}

\newcommand\Ev{\lm_1^\aa(\Omg,y)}
\newcommand\Ef{u_1^\aa}
\newcommand\EvB{\lm_1^\aa(\cB,\bfo)}

\newcommand\dl{\delta}
\newcommand\nb{\nabla}

\newcommand{\comm}[1]{}

\renewcommand\aa{\alpha}
\newcommand\lm{\lambda}
\newcommand\s{\sigma}
\newcommand\p{\partial}

\newcommand{\spn}{\mathrm{span}\,}
\newcommand\frh{\mathfrak{h}}
\newcommand\eps{\varepsilon}

\newcommand\arr{\rightarrow}

\newcommand\dd{{\mathsf{d}}}

\newcounter{counter_a}
\newenvironment{myenum}{\begin{list}{{\rm(\roman{counter_a})}}%
{\usecounter{counter_a}
\setlength{\itemsep}{1.ex}\setlength{\topsep}{0.8ex}
\setlength{\leftmargin}{5ex}\setlength{\labelwidth}{5ex}}}{\end{list}}

\newcounter{counter_b}

\numberwithin{figure}{section}
\numberwithin{equation}{section}
\theoremstyle{plain}% default
\newtheorem*{thm*}{Theorem}
\newtheorem{thm}{Theorem}[section]

\newtheorem{lem}[thm]{Lemma}
\newtheorem{prop}[thm]{Proposition}

\newtheorem{cor}[thm]{Corollary}

\theoremstyle{remark}
\newtheorem{remark}[thm]{Remark}
\theoremstyle{plain}

\newcommand\ov{\overline}
\newcommand\wt{\widetilde}

\def\ov{\overline}
\def\frs{{\mathfrak s}}
\def\frt{{\mathfrak t}}

      \def\dC{{\mathbb C}}

      \def\dR{{\mathbb R}}

   \def\sfH{{\mathsf H}}

      \def\sfR{{\mathsf R}}
\def\sfS{{\mathsf S}}   \def\sfT{{\mathsf T}}

\def\cA{{\mathcal A}}   \def\cB{{\mathcal B}}   
\def\cD{{\mathcal D}}      
   \def\cH{{\mathcal H}}   
      \def\cL{{\mathcal L}}

\def\cS{{\mathcal S}}

\newcommand{\dom}{\mathrm{dom}\,}

\makeatletter
\def\section{\@startsection{section}{1}\z@{.9\linespacing\@plus\linespacing}%
	{.7\linespacing} {\fontsize{13}{14}\selectfont\bfseries\centering}}
\def\paragraph{\@startsection{paragraph}{4}%
	\z@{0.3em}{-.5em}%
	{$\bullet$ \ \normalfont\itshape}}

\@addtoreset{equation}{section}
\makeatother

\newtheorem{lemma}{Lemma}[section]

\title[Faber-Krahn inequalities with interactions]{\textsc{Faber-Krahn inequalities for Schr\"odinger operators with point and with Coulomb interactions}}

\author[V. Lotoreichik]{Vladimir Lotoreichik}
\address{
Department of Theoretical Physics\\
Nuclear Physics Institute, Czech Academy of Sciences, 
25068 \v{R}e\v{z}, Czech Republic\\
E-mail: {lotoreichik@ujf.cas.cz }
}

\author[A. Michelangeli]{Alessandro Michelangeli}
\address{Institute for Applied Mathematics \\ and Hausdorff Center for Mathematics \\ University of Bonn, \\ Endenicher Allee 60 
D-53115 Bonn, Germany
	E-mail: {michelangeli@iam.uni-bonn.de}
}

\begin{document}

\begin{abstract}
 We obtain new Faber-Krahn-type inequalities for certain perturbations of the Dirichlet Laplacian on a bounded domain. First, we establish a two- and three-dimensional Faber-Krahn inequality for the Schr\"odinger operator with point interaction: the optimiser is the ball with the point interaction supported at its centre. Next, we establish three-dimensional Faber-Krahn inequalities for one- and two-body Schr\"odinger operator with attractive Coulomb interactions, the optimiser being given in terms of Coulomb attraction at the centre of the ball. The proofs of such results are based on symmetric decreasing rearrangement and Steiner rearrangement techniques; in the first model a careful analysis of certain monotonicity properties of the lowest eigenvalue is also needed.
\end{abstract}

\maketitle	

\section{Background and outline}\label{sec:intro}

In this work we produce two types of generalisations of a famous, one century old, optimisation result due to G.~Faber~\cite{Faber-1923} and E.~Krahn~\cite{Krahn-1925}. Whereas our applications concern two distinct operators of interest, the conceptual scheme of the proofs and the technical tools utilised are similar in both cases, which is the reason we should like to present both results on the same footing.

In its original formulation Faber-Krahn inequality states that amongst all domains $\Omega\subset\mathbb{R}^d$, $d\geq 2$, with the same given finite volume, the lowest (principal) eigenvalue of the negative Dirichlet Laplacian is minimised by the ball. This is an archetypal result in the vast and ever growing field of variational methods for eigenvalue approximation and spectral optimisation. The main tool in the proof is the symmetric decreasing rearrangement -- see the monographs \cite{Bandle-book-1980,Kesavan-symm-2006,Baernstein-symm-2019} and the references therein.

The general concept underlying Faber-Krahn-type inequalities is the relation between geometry and spectral properties, an idea that can be traced back to Lord Rayleigh's celebrated conjecture \cite{Rayleigh-sound-1894}.

Historically, the first natural focus in this respect is precisely the Laplacian with Dirichlet boundary conditions on $\Omega$ (by `Laplacian' we shall understand henceforth the differential operator $-\Delta$, the minus sign being included in order to realise the operator as lower semi-bounded on $L^2(\Omega)$). Other settings of interest in the spirit of spectral optimisation, which can mean both the upper and the lower bound depending on the problem, are the Neumann Laplacian, whose lowest non-trivial eigenvalue is maximised by the ball \cite{Szego-1954,Weinberger-1956}, Laplacians with Robin \cite{Daners-2006,Daners-Kennedy-2007,Freitas-2015,Bucur-Freitas-Kennedy-2017Robin,Keady-W-2018,BFNT-2018,Laugesen-2019Robin,Bucur-Giacomini-2019}, Wentzell-Robin \cite{Kennedy-2008,Kennedy-2010}, Stekloff \cite{Weinstock-1954,Brock-2001}, or mixed \cite{Kovaric-miedbc-2014} boundary conditions, and more generally for Laplace-Beltrami operators for domains in compact Riemannian manifolds with various boundary conditions \cite{Li-Yau-1980,Chavel-1984,Xu-1995,Khalile-Lotoreichik-2019,Savo-2020}, just to scratch the surface of a huge and branched out research field. Of course, all this come up with a variety of techniques that may differ substantially from the rearrangement scheme of the original Faber-Krahn inequality. 
%\vl{\sout{for example, in the Neumann case a simple symmetrisation cannot work, as the first non-trivial Neumann eigenfunction has a change of sign, and in the Robin case the very useful property of domain monotonicity of the Dirichlet eigenvalues is lost and the level sets of the eigenfunctions may touch the boundary, thus preventing a straightforward generalisation of the classical proof of Faber-Krahn inequality.}}
%\marginpar{These are apparently not the main reasons for impossibility to apply the rearrangement technique.}

Beside investigating different sorts of boundary conditions, also other sorts of differential operators have been studied which give rise to Faber-Krahn-type inequalities, significantly $p$-Laplacians \cite{Bhattacharya-1999,Bucur-Daners-2010,Kennedy-2010,Dai-Fu-2011,Hu-Dai-2014}, magnetic Laplacians \cite{Fournais-Helffer-2019}, and Dirac operators \cite{Benguria-Fournais-Stock-vDB-2017,ABLOB-2020}. Certain analogous spectral optimisation results have been established also for Robin Laplacians on the exterior of compact sets \cite{Kreicik-Lotoreichik-2018-extCompactSets,Kreicik-Lotoreichik-2020-extCompactSets-II,Exner-Lotoreichik-2020} and on unbounded cones \cite{Khalile-Lotoreichik-2019}.

The direction we are concerned with here is the emergence of Faber-Krahn-type inequalities for suitable \emph{perturbations} of the Dirichlet Laplacian on bounded domains $\Omega$.

The first playground one may think of are of course Schr\"{o}dinger operators $-\Delta+V$ for suitable measurable potential $V:\Omega\to\mathbb{R}$. For instance (see, e.g., \cite[Sect.~4]{Benguria-Linde-Loewe-2011}) a straightforward adaptation of Faber-Krahn inequality holds, stating that amongst all $\Omega$'s with same finite volume, the lowest eigenvalue of $-\Delta+V$ with $V$ non-negative in $L^1(\Omega)$ and with Dirichlet boundary conditions always exceeds the lowest eigenvalue of the analogous Schr\"{o}dinger operator on the ball with potential given by the symmetric increasing rearrangement of $V$. In this framework, optimisation (say, of the first eigenvalue, of the fundamental gap, etc.) has been investigated with respect to various classes of potentials at fixed $\Omega$ (see, e.g., the survey in \cite[Chapter 8]{Henrot-Extremum-Problems-2006} and the references therein)

Here our focus splits into two lines. Dirichlet boundary conditions shall be assumed throughout. On the one hand, we are concerned with the question of optimising the lowest eigenvalue of a Schr\"{o}dinger operator with localised impurity. There are significant precursors \cite{CGIKO-2000,Harrell-Kroeger-Kurata-2001,Cupini-Vecchi-2019}, which qualified the optimal placement of an obstacle or a well within a fixed bounded domain $\Omega$, meaning, a positive or negative bump-like potential $V$ supported inside $\Omega$. We push this line further, by modelling the impurity with an operator of point interaction, that is, a singular delta-like perturbation of the Dirichlet Laplacian supported at a point $y\in\Omega$. %(A parallel line of investigation is active, on the optimisation with respect to singular perturbations supported on arcs \cite{Lotoreichik-arc-2019})
The connection with bump-like potentials of finite size is clear by analogy with the case of a point interaction Hamiltonian on the whole $\mathbb{R}^d$, $d\in\{1,2,3\}$: the latter can be indeed constructed as a suitable limit of Schr\"{o}dinger operators $-\Delta+V_n$ on $L^2(\mathbb{R}^d)$ along a sequence of sufficiently localised and regular potentials $V_n$ shrinking and spiking up to a delta-like profile as $n\to\infty$ \cite{AHK-1981-JOPTH}.

We thus consider a bounded domain $\Omg\subset\dR^d$ with $C^\infty$-boundary, $d\in\{2,3\}$, and the operator $\Op$ on $L^2(\Omega)$ for given $y\in\Omega$ and $\alpha\in\mathbb{R}$, namely the self-adjoint operator of point interaction supported at $y$ and with inverse scattering length $\alpha$. Loosely speaking, $\Op$ corresponds to the formal differential expression 
\[
	-\Delta + \nu_\alpha\delta_y
\]
with Dirichlet boundary condition at $\p\Omg$ and some coupling $\nu_\alpha$ (of which $\alpha$ is a suitable renormalisation). In fact, $\Op$ is rigorously defined as a self-adjoint extension in $L^2(\Omega)$ of the Dirichlet Laplacian restricted to smooth functions vanishing on neighbourhoods of $y$, a construction obtained in \cite{Blanchard-figari-Mantile-2007} (see also~\cite{CdV82}). Basic spectral properties and an amount of further results on $\Op$ were established in~\cite{Blanchard-figari-Mantile-2007,Exner-Mantile-2007,Posilicano-2013}. As the above singular perturbation at $y$ does not alter the lower semi-boundedness of the unperturbed Dirichlet Laplacian, it still makes sense to investigate the principal eigenvalue $\lm_1^\aa(\Omg,y)$ of $\Op$. In particular, \cite{Exner-Mantile-2007} proved strict monotonicity of $\lm_1^\aa(\Omg,y)$ with respect to certain directions along which $y$ is moved, thus a first partial answer to the question where to locate a point interaction of given strength so as to minimise $\lm_1^\aa(\Omg,y)$.

Our first main result in the present analysis is the solution to a problem that merges the above question with the isoperimetric question for domains with the same volume, namely the problem of optimising $\lm_1^\aa(\Omg,y)$ with respect to a \emph{simultaneous} variation of the domain $\Omg$ and the point $y\in\Omg$.

We demonstrate the Faber-Krahn inequality
\begin{equation*}%\label{eq:main_intro}
	\lm_1^\aa(\cB,\bfo)\le \lm_1^\aa(\Omg,y), 
\end{equation*}
where $\cB\subset\dR^d$ is the ball of the same volume as $\Omg$ centred at the origin $\bfo\in\dR^d$. This is achieved by means of rearrangement techniques (for which we provide a concise survey in Section \ref{sec:preliminaries}), combined with an accurate analysis of certain crucial features of $\lm_1^\aa(\Omg,y)$ and its associated eigenfunction, in particular the monotonicity of the former with respect to $\alpha$, and some convenient representations of the latter (an analysis that we develop in Sections \ref{sec:DirichletLaplacian} and \ref{sec:lowesteigenvalue}). We present the proof of this first main result in Section \ref{sec:FaberKrahnForHDpoint}. 
It is worth emphasizing that we demonstrate the above Faber-Krahn inequality following two alternative routes: a general one that relies on certain estimates available in the literature for Green functions for Dirichlet Laplacians on domains, and an additional one, applicable when $\lm_1^\aa(\Omg,y)>0$, that has the virtue of exploiting rearrangement techniques as for the original Faber-Krahn, and in fact allows us to re-prove independently the above mentioned Green function estimates. In this respect, it is remarkable that the latter are so intimately connected with the spectral theory of the Hamiltonian with a point interaction in bounded domain.

In the second line of investigation of this work, on the other hand, we are concerned with the optimisation of the lowest eigenvalue of certain three-dimensional Schr\"{o}dinger operators on bounded domain and with attractive Coulomb potential.

We actually examine two models. First, for generic bounded domain $\Omg\subset\dR^3$ with $C^\infty$-boundary, $y\in\Omega$, and $q>0$, we consider the operator
\[
 \sfT_{q,y}^\Omg=-\Delta_x -\frac{q}{|x -y|}
\]
in its natural self-adjoint realisation on $L^2(\Omega)$ with Dirichlet boundary conditions at $\partial\Omega$. Here too the Coulomb perturbation produces a lower semi-bounded operator, with lowest eigenvalue $\mu_1^q(\Omg,y)$. For the latter we establish the Faber-Krahn inequality
\[
 \mu_1^q(\cB,\bfo) \le \mu_1^q(\Omg,y)\,.
\]
Thus, the configuration with Coulomb attraction at the centre of the ball minimises the principal Schr\"{o}dinger-Coulomb eigenvalue among all domains with equal volume and generic interaction centre $y$.

Next, we examine the \emph{two-body} counterpart of the previous model. Whereas $\Op$ and $\sfT_{q,y}^\Omg$ above are naturally interpreted as Hamiltonians for one non-relativistic quantum particle confined in $\Omega$ and subject to an interaction (of contact or Coulomb type) centred at a fixed point $y$, we now study the quantum Hamiltonian for two particles confined in $\Omega$ and coupled among themselves by a two-body, isotropic, Coulomb attraction of intensity $q > 0$. The Hamiltonian of interest becomes
\[
 \sfT_q^\Omg=-\Delta_{x_1} - \Delta_{x_2} -\frac{q}{|x_1-x_2|},
\]
canonically realised as a self-adjoint operator on  $L^2(\Omg\tm\Omg)$ with Dirichlet boundary conditions at $\partial(\Omega\tm\Omg)$. Here one should think of $q$ as the product of the (absolute values of) the charges of the two particles. Again, $\sfT_q^\Omg$ is lower semi-bounded, with lowest eigenvalue $\nu_1^q(\Omg)$. In the optimisation of $\nu_1^q(\Omg)$ over all $\Omega$'s with the same volume we establish the lower bound
\[
 2\mu_1^{q/2}(\cB,\bfo)\le \nu_1^q(\Omg)\,.
\]
The l.h.s.~above formally expresses the sum of the lowest energy levels of two identical particles in the ball $\cB$, each of which evolves uncoupled from the other and is subject instead to a Coulomb attraction from the centre of the ball, where now the product of the charge of the centre and the charge of each particle is half of the original $q$.

Such analysis is carried on in Section \ref{sec:Coulomb}. The conceptual scheme for the one-body case goes along the same line as for our Faber-Krahn inequality for the one-body point interaction. Exporting that scheme to the two-body case requires the replacement of the standard symmetric rearrangement tools with the Steiner rearrangement (which is also concisely reviewed in Section \ref{sec:preliminaries}).

In conclusion, we provide three new non-trivial examples of variational eigenvalue estimates for Schr\"{o}dinger self-adjoint operators on bounded domains, the first two of which are in the form of novel Faber-Krahn-type inequalities.

As the models we considered here appear not to have had previous scrutiny as far as spectral optimisation is concerned, an amount of interesting open questions obviously arise, on some of which we are committed to, as counterparts of the corresponding analysis for the free Laplacian on $L^2(\Omega)$ with given boundary conditions of self-adjointness. This includes, for example, the question of uniqueness of the minimiser, or the problem of optimisation over a restricted class of domains with definite geometry (such as parallelepipeds with the same volume of a prescribed cube), or the behaviour with respect to different boundary conditions other than Dirichlet, or the optimisation of the fundamental gap, just to mention a few typical ones. It would be of interest also to supplement our analysis of $\sfT_{q,y}^\Omg$ and $\sfT_q^\Omg$ by including the case of negative coupling $q<0$, and by comparing $\nu_1^q(\Omg)$ with $\nu_1^q(\cB)$.

We believe that these attractive topics deserve future investigation.

\bigskip

\textbf{Notation.} Beside an amount of fairly standard notation, as well as further convenient shorthand that will be introduced in due time, we shall adopt the following conventions throughout.

${\;}$
\!\!\!\!\!\!\!\!\!\!\!\!\!\begin{tabular}{ccl}
 $|\cA|$ & & Lebesgue measure of a measurable set $\cA\subset\mathbb{R}^d$ \\
 $\partial\Omega$ & & boundary of a domain $\Omega\subset\mathbb{R}^d$ \\
 $\mathrm{supp}\,f$ & & support of the function $f$ \\
 $(\cdot,\cdot)_{L^2(\Omega)}$ & & $L^2$-scalar product, anti-linear in the first entry, linear in the second \\
  $\mathbbm{1}$ & & identity operator (on the space that will be clear from the context) \\
  $\mathrm{dom}$ & & domain of an operator or a quadratic form \\
 $\sigma(T)$ & & spectrum of the operator $T$ w.r.t.~the underlying Hilbert space \\
  $T_1\otimes T_2$ & & tensor products of operators $T_1,T_2$ (w.r.t.~the underlying Hilbert spaces)
\end{tabular}

Unless when it becomes relevant to emphasize that, we shall tacitly understand all identities $f=g$ between $L^2$-functions in the 	sense of almost everywhere identities.

\section{Preparatory materials}\label{sec:preliminaries}

\subsection{Symmetric decreasing rearrangement} Let us start with introducing the symmetric decreasing rearrangement and recalling some of its fundamental properties. This is standard material; we refer to the monographs~\cite{Bandle-book-1980,Lieb-Loss-Analysis,Kesavan-symm-2006,Leoni-2017-Sobolev} for additional details.

Let $\cA$ be a measurable set of finite
volume in the Euclidean space $\dR^d$ of dimension $d \ge 2$. Its \emph{symmetric rearrangement}
$\cA^*$ is the open
ball $\cB\subset\dR^d$ centred at the origin $\bfo\in\dR^d$ and such that
$|\cA| = |\cB|$.

Let $u\colon\dR^d\arr \dR$ be a non-negative measurable function
that vanishes at infinity, in the sense
that all its positive level sets have
finite measure:
\begin{equation}\label{eq:vanish}
	\big|\{x\in\dR^d\,\big|\, u(x) > t\}\big| < \infty,
	\qquad\forall\, t >0.
\end{equation}
We define the \emph{symmetric decreasing
rearrangement} $u^*$ of $u$ by symmetrizing
its level sets as
\begin{equation}
	u^*(x) := \int_0^\infty \chi_{\{u > t\}^*}(x)\dd t.
\end{equation}
Here $\chi_\cA\colon \dR^d\arr\dR$ is the characteristic function
of a measurable set $\cA\subset\dR^d$.

The rearrangement $u^*$ has a number of straightforward properties.

\begin{lem}\label{lem:rearr}
	Let $u\colon \dR^d\arr\dR$, $d \ge 2$, be a non-negative measurable function vanishing at infinity. Let $\cA\subset\dR^d$ be a measurable
	set of finite volume. Then:
	\begin{myenum}
	\item $u^*$ is non-negative;
	\item $u^*$ is radially symmetric and non-increasing;
	\item $u$ and $u^*$ are equi-measurable, i.e.,
	\[
	 |\{x\in\dR^d\,\big|\, u(x) > t\}| = |\{x\in\dR^d\,\big|\, u^*(x) > t\}|
	\]
	for all $t > 0$;
		\item $\supp u \subset\cA$ implies $\supp u^*\subset \cA^*$;
		\item $(u^*)^2 = (u^2)^*$.
	\end{myenum}
\end{lem}

Let us collect further standard properties of the symmetric decreasing rearrangement that we shall use throughout. 

\begin{prop}\label{prop:rearr0}\cite[Theorem 3.4 and Lemma 7.17]{Lieb-Loss-Analysis}
	Let $u,v\colon \dR^d\arr\dR$ be non-negative measurable functions vanishing at infinity. Then the following hold.
	\begin{myenum}
		\item $\|u\|_{L^2(\dR^d)} 
		=
		\|u^*\|_{L^2(\dR^d)}$ (conservation of $L^2$-norm).
		\item 
		$\int_{\dR^d} u(x)v(x)\dd x \le 
		\int_{\dR^d} u^*(x)v^*(x)\dd x$
		(Hardy-Littlewood inequality).
		\item If
		$\nb u \in L^2(\dR^d)$ exists in the
		sense of  distributions, then
		$\nb u^*$ has the same property
		\[
			\|\nb u\|_{L^2(\dR^d)} \ge
		\|\nb u^*\|_{L^2(\dR^d)}\qquad
		\text{(P\`{o}lya-Szeg\H{o} inequality).}
		\]
	\end{myenum}	
	In particular, $u\in H^1(\dR^d)$ implies that
	$u^*\in H^1(\dR^d)$ as well.
\end{prop}

In view of Lemma~\ref{lem:rearr}\,(iv),	
the operation of taking symmetric decreasing rearrangement can be naturally extended to functions
defined on domains.

Let $\Omg\subset\dR^d$ be a bounded domain
with $C^\infty$-smooth boundary.
For a non-negative measurable function $u\colon \Omg\arr\dR$, 
we denote by $\wt{u}\colon\dR^d\arr\dR$ its extension by zero to the whole $\dR^d$
and define the \emph{symmetric rearrangement $u^*$ of $u$} as
\[
u^* := \wt{u}^*|_{\Omg^*}\colon \Omg^*\arr\dR.
\]
In this respect Proposition \ref{prop:rearr0} has the following corollary.

\begin{cor}\label{prop:rearr}
	Let $\Omg\subset\dR^d$, $d \ge 2$, be a bounded domain with $C^\infty$-smooth boundary and let the ball $\cB := \Omg^*$ be its symmetric
	rearrangement. 	
	Let $u,v\colon \Omg\arr\dR$ be non-negative measurable functions. Then:
	\begin{myenum}
		\item $\|u\|_{L^2(\Omg)} 
		=
		\|u^*\|_{L^2(\cB)}$;
		\item 
		$\int_{\Omg} u(x)v(x)\dd x \le 
		\int_{\cB} u^*(x)v^*(x)\dd x$;
		\item if additionally
		$u \in H^1_0(\Omg)$, then
		$u^* \in H^1_0(\cB)$ and
		$\|\nb u\|_{L^2(\Omg)} \ge
			\|\nb u^*\|_{L^2(\cB)}$.
	\end{myenum}	
\end{cor}	

Next, following the lines of~\cite{Brascamp-Lieb-Luttinger-1974,Brock-Solynin-2000,Capriani-2014-Steiner}, we introduce for a non-negative measurable function $u = u(x_1,x_2)\colon\dR^{2d}\arr\dR$, $x_1,x_2\in\dR^d$,  $d \ge 2$  its \emph{Steiner rearrangements}
with respect to the first $d$ and the last $d$ variables, that is,
\begin{equation}\label{eq:Steiner}
 \begin{split}
  (S_1u)(\cdot,x_2) \;&=\; (u(\cdot,x_2))^* \\
  (S_2u)(x_1,\cdot) \;&=\; (u(x_1,\cdot))^*\,.
 \end{split}
\end{equation}
Here, we implicitly assume that $u(x_1,\cdot)$ and $u(\cdot, x_2)$ are vanishing at infinity for almost all $x_1\in\dR^d$ and $x_2\in\dR^d$, respectively.

These rearrangements have an amount of properties 
reminiscent of those for the standard symmetric decreasing
rearrangement and follow from the latter via simple arguments (see \cite[Theorem 8.2]{Brock-Solynin-2000}).

\begin{prop}
	\label{prop:rearr2}
	Let $u\in H^1(\dR^{2d})$ be real-valued and non-negative.
	Let $\cA\subset\dR^d$ be a measurable
	set of finite volume.
	Then
	\[
	 S_1u,\;S_2u, \;S_1S_2u, \;S_2S_1u\;\in\; H^1(\dR^{2d}).
	\]
Moreover, the following hold.
	\begin{myenum}
		\item If $\supp u \subset\cA\tm\cA$, then 
		$\supp (S_1u)\subset \cA^*\tm\cA$,
		$\supp (S_2u_2)\subset \cA\tm\cA^*$
		and $\supp (S_1S_2u), \supp(S_2S_1 u) \subset\cA^*\tm\cA^*$.
		\item 
		$(S_1u)^2 = S_1u^2$,
		$(S_2u)^2 = S_2u^2$,
		$(S_1S_2u)^2 = S_1S_2u^2$,
		and
		$(S_2S_1u)^2 = S_2S_1u^2$.
		\item If $u(x_1,x_2) = u(x_2,x_1)$ for a.e. $x_1,x_2\in\dR^d$, then $S_1S_2u = S_2S_1u$. 
		\item
		$\|u\|_{L^2(\dR^{2d})} 
		=
		\|S_1u\|_{L^2(\dR^{2d})}
		=
		\|S_2u\|_{L^2(\dR^{2d})}
		=
		\|S_1S_2u\|_{L^2(\dR^{2d})}
		=
		\|S_2S_1u\|_{L^2(\dR^{2d})}
		$.
		\item 
		$\|\nb u\|_{L^2(\dR^{2d})} \ge
			\|\nb S_1S_2u\|_{L^2(\dR^{2d})}$
		and	
		$\|\nb u\|_{L^2(\dR^{2d})}  \ge
		\|\nb S_2S_1u\|_{L^2(\dR^{2d})}$.
	\end{myenum}	
\end{prop}

\subsection{Green functions}
Let us now recall basic properties of  
the Green function associated with the Dirichlet Laplacian on a bounded smooth domain, focusing in particular on the connections between Green function and symmetric decreasing rearrangement \cite{Bandle-78}.

Let $\Omg\subset\dR^d$, $d\in\{2,3\}$, be a bounded domain
with $C^\infty$-smooth boundary. Consider
the self-adjoint Dirichlet Laplacian $\OpD$ in the Hilbert space $L^2(\Omg)$
\begin{equation}\label{eq:DLapl}
\begin{split}
 \dom\OpD \;&:=\; H^2(\Omega)\cap H^1_0(\Omega) \\
 \OpD u \;&:=\; -\Delta u\,.
\end{split}
\end{equation}
The
operator $\OpD$ is lower semi-bounded and with purely discrete spectrum, and its lowest eigenvalue 
$\lm_1(\Omg)$ is strictly positive. Let
\begin{equation}\label{eq:resolvent}
	\sfR_{\rm D}^\Omg(z) := \big(\OpD - z\big)^{-1},
	\qquad
	z\in \dC\sm\s(\OpD),
\end{equation}
be the resolvent of $\OpD$ at the point $z$. $\sfR_{\rm D}^\Omg(z)$ acts  on $L^2(\Omega)$ as a compact integral operator with kernel $G_z^\Omg \colon \ov\Omg\tm\ov\Omg\arr\dR$, called the \emph{Green function} associated with $\OpD$. In the sense of distributions
one has
\[
	((-\Dl - z)G_z^\Omg)(x,y) = \dl(x-y),
\]
where $\dl(\cdot)$ is the standard Dirac distribution in $\dR^d$, supported at the origin.

We will also need to refer to the Green function $\cS$ of the free Laplacian on $\dR^d$, defined in complete analogy to $G_z^\Omg$. In fact, $\cS$ is given explicitly by
\begin{equation}\label{eq:freeGreen}
	\cS(x,y) := F(|x-y|)
	\qquad\text{with}\qquad
	F(t) :=  
	\begin{cases}
		-\frac{\ln t}{2\pi},   &  d = 2,\\
		\frac{1}{4\pi t},     &  d = 3.
	\end{cases}
\end{equation}
Correspondingly, we define
\begin{equation}\label{eq:HzOmegaxy}
	 \cH^\Omg_z(x,y) := G_z^\Omg(x,y) - \cS(x,y)\,.
	\end{equation}

\begin{prop}\cite[Sect.~1]{Bandle-78}
\label{prop:Green1}
Let $\Omega\subset\dR^d$, $\lm_1(\Omg)$, and $G_z^\Omg$ be as above.
Assume that $z\in (-\infty,\lm_1(\Omg))$.
Then:
\begin{myenum}
	\item for fixed $y\in\Omg$, the function $x\mapsto \cH^\Omg_z(x,y)$ is continuous
	on $\ov{\Omg}$, whence also
\begin{equation}\label{eq:Green}
	G_z^\Omg(x,y) = \cS(x,y) + \cH^\Omg_z(y,y) + o(1)
	\qquad x\arr y\,;
\end{equation}	
	\item $G_z^\Omg(x,y) = 0$ for any $x\in\p\Omg$
	and all $y\in \Omg$;
	\item $G_z^\Omg$ is positive in $\Omg\tm\Omg$.
\end{myenum}
\end{prop}

In what follows,
\[
 \cB \equiv \cB_R := \{x\in\dR^d\,\big|\, |x| < R\}
\]
stands for the ball centred at $\bfo\in\dR^d$ and being such that $|\Omg| = |\cB|$. 
We shall also use the shorthand
\begin{equation}\label{eq:g}
   \begin{split}
    \hzy &:= \cH^\Omg_z(y,y) \\
    \gzy &:= G_z^\Omg(\cdot,y)\colon\ov\Omg\arr\dR^+
   \end{split}
\end{equation}
for any fixed $z < \lm_1(\Omg)$ and $y\in\Omg$.

Here are relevant properties of the Green function $G_z^\Omg$ with respect to the symmetric decreasing rearrangement.

\begin{prop}\label{prop:Green2}
$\Omg\subset\dR^d$, $d\in\{2,3\}$, be a bounded domain
with $C^\infty$-smooth boundary.
	Let $z < \lm_1(\Omg)$ and $y\in\Omg$.
	Then:
	\begin{itemize}
	 \item[(i)] $\gzy \in L^2(\Omg)$\,;
	 \item[(ii)] $0\le (g_{0,y}^\Omega)^* \le g_{0,\bfo}^\cB$\,;
	 \item[(iii)] $\hzy \le h_{z,\bfo}^\cB$\,;
	 \item[(iv)] $h_{0,y}^\Omg < 0$ if $d = 3$\,.
	\end{itemize}
\end{prop}

\begin{proof}
	(i) Square-integrability of $\gzy$ is a consequence
	of Proposition~\ref{prop:Green1}\,(i) and of
	the fact that the function $\Omg\ni x\mapsto\cS(x,y)$ with fixed $y\in\Omg$ is square-integrable.
	
	(ii) follows directly from~\cite[Theorem 2.1 with $p=0$]{Bandle-78}. 
	
	(iii) Recall that $R >0$ is the radius of the ball $\cB$. By~\cite[Lemma 2.3]{Bandle-78} there is $R' \le R$ such that $h^{\cB_{R'}}_{z,\bfo} = h^\Omega_{z,y}$. According to~\cite[proof of Lemma 2.3]{Bandle-78} the function $r\mapsto h^{\cB_r}_{z,\bfo}$ is increasing. Hence, we conclude that 
	$h^{\cB}_{z,\bfo} \ge h^{\cB_{R'}}_{z,\bfo} = h^\Omega_{z,y}$. 
	
	(iv) follows  from (iii) (with $z = 0$) and from $h_{0,y}^\cB < 0$ (see~\cite[\S II.2.2]{Bandle-book-1980}).
\end{proof}

Denoting by $F^{-1}$ the inverse of the function $F$ in~\eqref{eq:freeGreen}, we set 
\begin{equation}\label{eq:conformal_radius}
	R^\Omg_y := F^{-1}(-h_{0,y}^\Omg) > 0\,.
\end{equation}
The quantity $R^\Omg_y$
is actually the \emph{conformal radius} of $\Omg$ at the point $y$ when $d = 2$, or the \emph{harmonic
radius} of $\Omg$ at $y$ when $d = 3$. (We refer to \cite{Bandle-Flucher-1996} and the references therein for a detailed discussion on conformal and harmonic radii.) In the case $z = 0$, Proposition~\ref{prop:Green2}\,(iii)
reduces to the inequality $R_y^\Omg \le R_\bfo^\cB$
for all $y\in\Omg$.

\section{Dirichlet Laplacian on bounded domain with point interaction}\label{sec:DirichletLaplacian}

In this Section we review the construction of the Dirichlet Laplacian
on a bounded domain with a point interaction and we collect an amount of relevant properties. This combines two complementary languages: the operator theoretic self-adjoint extension scheme and the quadratic form approach.

Let in the following $\Omg$ be a bounded domain in $\dR^d$ with $C^\infty$-boundary, $d\in\{2,3\}$, and let the point $y\in\Omg$ be fixed.

It is standard to see (see, e.g., \cite[Lemma 1]{Blanchard-figari-Mantile-2007}) that 
\begin{equation}\label{eq:S}
\begin{split}
 	\dom \sfS &:= \big\{ u\in H^2(\Omg) \cap H^1_0(\Omg) \,\big|\, u(y) = 0 \big\} \\
 	\sfS u &:= - \Delta u
\end{split}
\end{equation}
is a densely defined, closed, symmetric operator on $L^2(\Omg)$, with lower bound $\lm_1(\Omg)$ (the strictly positive, lowest eigenvalue of the Dirichlet Laplacian $\OpD$ from \eqref{eq:DLapl}), with deficiency indices $(1,1)$ and with deficiency subspace
\begin{equation}
 \ker (\sfS^*-z) \;=\; \spn\,\{\gzy\}\,,\qquad z < \lm_1(\Omg)\,,
\end{equation}
where $\gzy$ is the Green function \eqref{eq:g}. It is also standard to see that $\OpD$, obviously a self-adjoint extension of $\sfS$, is precisely the Friedrichs extension of $\sfS$ (the operator domain of the former is contained in the form domain of the latter, namely in the $H^1$-closure of $\dom \sfS$).

As a consequence of these facts and of the Vi\v{s}ik-Birman decomposition formula (see, e.g., \cite[Theorem 1]{GMO-KVB2017}), any $u\in\dom \sfS^*$ decomposes as
\begin{equation}\label{eq:u_in_adjdom}
	u = u_0 + c_1 \RzD \gzy + c_0 \gzy,
		\qquad
	c_0,c_1\in\dC,\,~u_0 \in \dom\sfS,									
\end{equation}
for any $z < \lm_1(\Omg)$, where $\RzD$ is the resolvent~\eqref{eq:resolvent} of $\OpD$. At each fixed parameter $z$, the decomposition \eqref{eq:u_in_adjdom} is unique in terms of the $u$- and $z$-dependent elements $c_0,c_1,u_0$, and
\begin{equation}\label{eq:action}
	\sfS^* u 
	= 
	\sfS u_0 + c_1 z\RzD \gzy + (c_1 + c_0z) \gzy.
\end{equation}

The self-adjoint extensions of $\sfS$ form a one-real-parameter family $\{\sfS_\beta\,|\beta\in\mathbb{R}\}\cup\{\OpD\}$, where
\begin{equation}\label{eq:Op}
\begin{split}
 \dom\sfS_\bb \;&:=\; 
	\Big\{
	 u = u_0 + c\big(\bb\RzD \gzy + \gzy\big)\,\big|\,
	 u_0 \in\dom\sfS, c\in\dC\Big\} \\
	 \sfS_\bb u  \;&:=\; \sfS^*u\,.
\end{split}
\end{equation}
Formula \eqref{eq:Op} is a direct application to the present unit-deficiency-index case of the general classification formula for the self-adjoint extensions of a lower semi-bounded (and densely defined) symmetric operator (see, e.g., \cite[Theorem 5]{GMO-KVB2017} or \cite[Sect.~14.8]{schmu_unbdd_sa}).
The Friedrichs extension $\OpD$ formally corresponds to $\beta=\infty$.

The extensions $\sfS_\beta$ are equivalently characterised in terms of their quadratic forms: the form $\frs_\beta$ associated with each $\sfS_\beta$ is given by
\begin{equation}\label{eq:form_s}
\begin{split}
	\dom\frs_\bb \;&=\;\big\{u = v + \xi \gzy\,\big|\, v \in H^1_0(\Omg), \xi\in\dC \big\} \\
	\frs_\bb[u] \;&=\;  \|\nb v\|^2_{L^2(\Omg)}
	+ \bb\|\gzy\|^2_{L^2(\Omg)}\cdot|\xi|^2 + 
	z\big(\|u\|^2_{L^2(\Omg)}-\|v\|^2_{L^2(\Omg)}\big)
\end{split}
\end{equation}
(see, e.g., \cite[Theorem 7]{GMO-KVB2017} or \cite[Theorem 14.24]{schmu_unbdd_sa}).

The combination of the asymptotics \eqref{eq:Green} for $\gzy$ with the decomposition \eqref{eq:Op} implies that 
any function $u \in \dom \sfS_\bb$ behaves in the vicinity of the point $y$ as 
\begin{equation}\label{eq:expansion1}
	u(x) = 
	c\Big(
		\bb\|\gzy\|^2_{L^2(\Omg)}+ \cS(x,y) + 	\hzy 
	\Big) + o(1),	
	\qquad x\arr y. 
\end{equation}
This short-scale behaviour, as argued already in \cite[Sect.~III]{Blanchard-figari-Mantile-2007}, has the form 
\begin{equation}\label{eq:expansion2}
	u(x) = c\left(\cS(x,y) + \aa + o(1)\right),\qquad x\arr y
\end{equation}
which is typical of the low-energy scattering of a quantum particle over a scattering centre with zero-range interaction and with $s$-wave scattering length $(-\alpha)^{-1}$, in suitable units, as originally identified by  Bethe and Peierls \cite{Bethe_Peierls-1935} (whence also the nomenclature of \emph{Bethe-Peierls contact condition} -- see, e.g., \cite[Sect.~2]{MO-2016}). It is therefore meaningful to re-parametrise the $\sfS_\bb$'s  in terms of the new, physical grounded extension parameter
\begin{equation}\label{eq:alphabeta}
	\aa	= \bb\|\gzy\|^2_{L^2(\Omg)} +\hzy.
\end{equation}

Upon plugging \eqref{eq:alphabeta} into \eqref{eq:Op} and \eqref{eq:form_s}, we can summarise the above considerations as follows.

\begin{prop}\label{prop:OP} 
Let $\Omg$ be a bounded domain in $\dR^d$ with $C^\infty$-boundary, $d\in\{2,3\}$, and let $y\in\Omg$. Correspondingly, let $\sfS$ be as in \eqref{eq:S}.
\begin{itemize}
 \item[(i)] The self-adjoint extensions of $\sfS$ in $L^2(\Omega)$ constitute the one-real-parameter family
 \[
  \big\{\Op\,\big|\,\alpha\in\mathbb{R}\big\}\cup\{\OpD\}
 \]
 with each element $\Op$ given, fixed $z < \lm_1(\Omg)$, by
 \begin{equation}\label{eq:HalphayD}
  \begin{split}
   \dom\Op\;&=\;
   \left\{
   \begin{array}{c}
    u = u_0 + 
		c\|\gzy\|^{-2}_{L^2(\Omg)}
		\big(\aa - \hzy\big)\RzD\gzy + c\gzy \\
    \textrm{for some }\;u_0 \in \dom\sfS,
		c\in\dC
   \end{array}
   \right\} \\
   \Op u\;&=\;-\Dl u_0 + 
		c\|\gzy\|^{-2}_{L^2(\Omg)}\big(\aa - \hzy\big)
		\big[z\RzD\gzy + \gzy\big] +c z\gzy\,.
  \end{split}
 \end{equation}
 \item[(ii)] The quadratic form $\frh_{\aa,y}^\Omg$ of $\Op$ is given, fixed $z < \lm_1(\Omg)$, by
 \begin{equation}\label{eq:form}
	\begin{split}
		\dom \frh_{\aa,y}^\Omg  \;&=\; \big\{u = v + \xi 	g_{z,y}^\Omg\,|\, v \in H^1_0(\Omg), \xi\in\dC \big\} \\
		\frh_{\aa,y}^\Omg[u]  \;&=\; \|\nb v\|^2_{L^2(\Omg)} + 
		\Big(\aa - \hzy\Big)	|\xi|^2 + z\big(\|u\|^2_{L^2(\Omg)} - \|v\|^2_{L^2(\Omg)}\big)\,.
	\end{split}
	\end{equation}
\end{itemize}
 Fixed the parameter $z$, the decompositions in \eqref{eq:HalphayD} and \eqref{eq:form} of $u$ are unique. $\OpD$ is the Friedrichs extension.
\end{prop}

We shall refer to each $\Op$ as a \emph{self-adjoint Dirichlet Laplacian with point interaction on $\Omega$ with interaction centre $y\in\Omega$ and interaction strength $(-\aa)^{-1}$}.

It is convenient to introduce further shorthand notation:
\begin{equation}\label{eq:brevity}
 \begin{split}
  \gy \;&=\; g_{0,y}^\Omg \\
  \hy \;&=\; h_{0,y}^\Omg \\
  \frm \;&=\; \frh_{\aa,0,y}^\Omg \\
  \ggy \;&=\; \|\gy\|^{-2}_{L^2(\Omg)}\,.
 \end{split}
\end{equation}
This allows one to re-write \eqref{eq:HalphayD} and \eqref{eq:form}, with the choice $z=0$, respectively as
\begin{equation}\label{eq:HalphayD-2}
\begin{split}
	\dom\Op \;&=\;\left\{
   \begin{array}{c}
   u = u_0 + 
	c\ggy\big(\aa - \hy\big)
	\sfR_{\rm D}^\Omg(0)\gy + c\gy \\
    \textrm{for some }\;u_0 \in \dom\sfS,
		c\in\dC
   \end{array}
   \right\} \\
	\Op u \;&=\; -\Dl u_0 + 
	c\ggy\big(\aa - \hy\big)
	\gy
\end{split}
\end{equation}
and
\begin{equation}\label{eq:h}
\begin{split}
\dom \frm  \;&=\; \cD_y^\Omg \;:=\; \big\{u = v + \xi 	g_{y}^\Omg\,|\, v \in H^1_0(\Omg), \xi\in\dC \big\} \\
	\frm[u]  
	\;&=\; \|\nb v\|^2_{L^2(\Omg)} + 
	\big(\aa -  \hy\big)|\xi|^2\,.
\end{split}
\end{equation}

Let us work out certain useful properties of the Hamiltonian $\Op$.

\begin{prop}\label{prop:basic}~

	\begin{myenum}
	\item $\Op$ is lower semi-bounded and has compact resolvent. 
	\item 
	$\Op \ge 0 $ if and only if $\aa \ge \hy$.
	\item The map $\aa\mapsto\Op$
	is a non-decreasing operator-valued function 
	in the sense of ordering of forms.
	\item The map $\aa\mapsto \Op$ is continuous in the strong resolvent sense.
	\item $\Op\to\OpD$ as $\aa \arr +\infty$ in the strong resolvent sense.
	\end{myenum}
\end{prop}
\begin{proof}
	(i) 
	Since $\Op$ and $\OpD$
	are both self-adjoint extensions of the symmetric operator $\sfS$, and $\sfS$ has unit deficiency indices, the difference of their resolvents is a rank-one operator.
	Hence, semi-boundedness of $\Op$ and compactness of its resolvent follow
	from respective properties of $\OpD$. 

	(ii) When $\aa \ge \hy$, non-negativity of $\Op$ follows
	from the fact that $\frm[u] \ge 0$ for all $u\in\dom\frm$, which can be seen in \eqref{eq:h}.	
	Conversely, when $\aa < \hy$, \eqref{eq:h} yields
	\[
		\frm[\gy] =	\aa-\hy < 0.
	\] 
	Hence, the min-max principle implies that the negative spectrum of $\Op$ is non-empty.
	
	(iii)
	The claimed property follows from the fact that $\mathrm{dom}\,\frm$ is independent of $\aa$ and that
	for any $u \in\cD_y^\Omg$ one has $\frh_{\aa_1,y}^\Omg[u] \le \frh_{\aa_2,y}^\Omg[u]$ whenever $
	\aa_1 \le \aa_2$.
	
	(iv)
	Continuity of the operator-valued
	function $\aa\mapsto \Op$ in the strong resolvent sense is a consequence of continuity of the scalar-valued function 
	$\aa\mapsto\frm [u]$ for any $u\in \cD_y^\Omg$ combined with~\cite[Theorem XIII.3.6]{Kato-perturbation}.
	
	(v) The claim (iii), combined with the monotone convergence theorem for quadratic forms~\cite[Theorem S.14]{rs1}  imply $\Op\arr\OpD$ as $\aa\arr+\infty$ in the strong resolvent sense.
\end{proof}

\section{The lowest eigenvalue and the ground state of $\Op$}\label{sec:lowesteigenvalue}

In this Section we discuss the properties of the lowest eigenvalue $\Ev$ of the Hamiltonian $\Op$, and of the corresponding eigenfunction.

For our later purposes, crucial features to analyse are the dependence of such objects (and associated quantities) on the extension parameter $\aa$, as well as the convenient representations of the ground state eigenfunction.

We start with deriving a first set of results in this spirit.

\begin{prop}\label{prop:EVEF}
Let $\Omg$ be a bounded domain in $\dR^d$ with $C^\infty$-boundary, $d\in\{2,3\}$, and let $y\in\Omg$, $\alpha\in\mathbb{R}$. Correspondingly, let $\Op$ be as in Proposition \ref{prop:OP}, and let $\Ev$ be its lowest eigenvalue. (As a reference, let us recall that $\lm_1(\Omega) > 0$ is the lowest eigenvalue of the Dirichlet Laplacian $\OpD$ from \eqref{eq:DLapl}). One has:
	\begin{itemize}
	\item[(i)] $\Ev < \lm_1(\Omg)$;
	\item[(ii)] the spectrum of $\Op$ in the interval $(-\infty,\lm_1(\Omg))$ consists of a unique simple eigenvalue;		%
	\item[(iii)] the function $\dR\ni\aa\mapsto\Ev$ is continuous;		%
	\item[(iv)] $\displaystyle\lim_{\aa\arr-\infty} \Ev = -\infty$
	\;and\;
	$\displaystyle\lim_{\aa\arr +\infty} \Ev = \lm_1(\Omg)$.
	\end{itemize}
	Let now $\Ef$ be the eigenfunction of $\Op$ corresponding to $\Ev$ (up to a multiplicative constant). Then:
	\begin{itemize}
	\item[(v)] one has the representation
	\[
	 \Ef \;=\; c \gzy\qquad\textrm{ with }\;z \,=\, \Ev\;\textrm{ and }\;c \in\dC\sm\{ 0\}\,,
	\]
	hence $u_1^\aa$ can be chosen to be positive on $\Omg$;	
	\item[(vi)] when in particular $\Ev > 0$, one has the representation
	\[
	 u_1^\aa \;=\; v + \xi \gy\qquad\textrm{ with }v\in H^1_0(\Omg)\,,\;\textrm{$v \ge 0$ on $\Omg$}\,,\;\textrm{ and }\;\xi \ge 0\,.
	\]
	\end{itemize}
\end{prop}
\begin{proof}
	(i)
	Let $u_1\in H^1_0(\Omg)\cap H^2(\Omg)$ 	be the ground state eigenfunction of $\OpD$. Without loss of generality, one can assume that $u_1$ is positive on $\Omg$. 
	Consider  a family of test functions
	\[
	 u_{1,\eps} \;:=\; u_1 + \eps \gy\,,\qquad \eps \in\dR\,.
	\]
	Owing to \eqref{eq:h}, $u_{1,\eps}\in \cD_y^\Omg$. 	
	Differentiating the scalar-valued function
	$\cL(\eps) := \frm[u_{1,\eps}] - \lm_1(\Omg)\|u_{1,\eps}\|^2_{L^2(\Omg)}$ 
	at $\eps = 0$ gives
	\[
		\cL'(0)
		\;=\;	
		-2\lm_1(\Omg)\int_\Omg u_1 \gy \dd x\,.
	\]
	As $\gy$ is positive on $\Omg$ (Proposition \ref{prop:Green1}\,(iii) and \eqref{eq:g}), we conclude that $\cL'(0) < 0$ and the min-max principle yields the inequality $\Ev < \lm_1(\Omg)$.
		
	(ii) As argued already for the proof of Proposition \ref{prop:basic}\,(i), $\Op$ and $\OpD$ differ in the resolvent sense by a rank-one operator. %~\cite[\S 9.3, Lem. 2]{Birman-Solomjak-1980book}.
	From this, and from the fact that $\inf\s(\OpD)= \lm_1(\Omg)$, one can deduce \cite[\S 9.3, Theorem 3]{Birman-Solomjak-1980book} that the rank of the spectral projection for $\Op$ corresponding to the interval $(-\infty,\lm_1(\Omg))$ is either $0$ or $1$.
	Taking into account that $\Ev < \lm_1(\Omg)$, we eventually conclude that this rank equals to one,
	which is equivalent to the claim.

	(iii) 	The continuity of $\Ev$ with respect to $\aa$ directly follows from the spectral convergence result \cite[Theorem VIII.1.14]{Kato-perturbation} and from Proposition \ref{prop:basic}\,(iv).

	(iv) 
	By the min-max principle and \eqref{eq:h},
	\[
		\Ev	
		\;\le\; 
		\frac{\frm[\gy]}{\|\gy\|_{L^2(\Omg)}^2}
		\;=\;
		\ggy
		\big(\aa  -\hy\big)\xrightarrow[]{\;\alpha\to -\infty\;}\;-\infty\,,
	\]
	whence the first of the two claimed limits follows. The second limit is a consequence of the strong resolvent convergence in Proposition \ref{prop:basic}\,(v) and the spectral convergence result \cite[Theorem VIII.1.14]{Kato-perturbation}.
	
	(v) For the proof of this part and of the next one, let us switch to the shorthand $\lm_1 = \Ev$.
	As obviously $\Ef\in\mathrm{dom}\sfS^*$, then
	\[
	 \Ef \;=\; u_0 + c_1\sfR_{\rm D}^\Omg(\lm_1)g_{\lm_1,y}^\Omg +c_0 g_{\lm_1,y}^\Omg
	\]
       for suitable $u_0,c_0,c_1$ (owing to \eqref{eq:u_in_adjdom} above), whence
	\[
	\begin{split}
	 (\Op - \lm_1)\Ef \;&=\; \big(\sfS^* - \lm_1\big)
		\big(u_0 + c_1 \sfR_{\rm D}^\Omg(\lm_1) g_{\lm_1,y}^\Omg\big) \\
		&=\;(\OpD-\lambda_1)\big(u_0 + c_1 \sfR_{\rm D}^\Omg(\lm_1) g_{\lm_1,y}^\Omg\big)\,.
	\end{split}
	\]
	Now, $\Op\Ef=\lm_1\Ef$ implies $u_0 + c_1 \sfR_{\rm D}^\Omg(\lm_1) g_{\lm_1,y}^\Omg=0$, because $\lm_1 < \inf\s(\OpD)$. Therefore, $\Ef = c_0 g_{\lm_1,y}^\Omg$.
	Moreover, the choice $c_0>0$ yields a positive $\Ef$, as $g_{\lm_1,y}^\Omg > 0$ (Proposition \ref{prop:Green1}\,(iii)).
	
	(vi) By assumption $\lm_1 > 0$ and therefore (Proposition \ref{prop:basic}\,(ii)) $\aa \ge \hy$.
	As $\Ef\in \dom\Op \subset \dom \frm$, then
	\[
		\Ef \;=\; v + \xi \gy
	\]
	for some $v\in H^1_0(\Omg)$ and $\xi\in\dC$ (owing to \eqref{eq:h} above). It is not restrictive to assume that $\xi \ge 0$ and that consequently (based on (v)) $v$ is real-valued. It remains to show that $v \ge 0$.
	To this aim, we pick the test function $u := |v| + \xi \gy$.
	One has
	\begin{equation}\label{eq:uu}
	\begin{split}
		\|u\|^2_{L^2(\Omg)} 
		& = 
		\|v\|^2_{L^2(\Omg)} +2\xi(|v|,\gy)_{L^2(\Omg)}
			+ |\xi|^2\|\gy\|^2_{L^2(\Omg)} \\
		& \ge
		\|v\|^2_{L^2(\Omg)} +2\xi (v,\gy)_{L^2(\Omg)}
		+ |\xi|^2\|\gy\|^2_{L^2(\Omg)} 
		= \|\Ef\|^2_{L^2(\Omg)},
	\end{split}
	\end{equation}
	where $\gy > 0$ was used (owing to \eqref{eq:g}, \eqref{eq:brevity}, and Proposition \ref{prop:Green1}\,(iii)). Then 
	\[
		\lm_1 \;\le\; \frac{\frm[u]}{\|u\|^2_{L^2(\Omg)}}
		\;=\;
		\frac{\|\nb v\|^2_{L^2(\Omg)} + 
			\Big(\aa - \hy\Big)
			|\xi|^2}{\|u\|^2_{L^2(\Omg)}}
		\;\le\;
		\frac{\frm[\Ef]}{\|\Ef\|^2_{L^2(\Omg)}} \;=\; \lm_1\,,
	\]
	having applied the min-max principle in the first step, the identity $\|\nb |v|\|_{L^2(\Omg)} = \|\nb v\|_{L^2(\Omg)}$ in the second, and the inequalities $\aa \ge \hy$ and~\eqref{eq:uu} in the third. The eigenvalue $\lm_1$ being simple, then necessarily $\Ef = u$, and thus $v = |v|$ is non-negative.	
\end{proof}

\begin{remark}\label{rem:bottom_always_below}
 Proposition \ref{prop:EVEF} shows that when passing from the Dirichlet Laplacian $\OpD$ to any of its \emph{singular perturbations} $\Op$, an eigenvalue is always created below the threshold $\lambda_1(\Omega)$ \emph{irrespective of the sign and magnitude of $\alpha$}. Thus, each self-adjoint extension $\Op$ of the symmetric operator $\mathsf{S}$ defined in \eqref{eq:S} has bottom $\inf\sigma(\Op)=\Ev$ strictly below the bottom of the Friedrichs extension of $\mathsf{S}$, i.e., below $\inf\sigma(\OpD)=\lambda_1(\Omega)$. Or, in other words, there are no other extensions besides the Friedrichs one with the same lower bound of $\mathsf{S}$. The above behaviour is not generic. For instance, singular perturbations of the self-adjoint Laplacian on $\mathbb{R}^d$ when $d=3$ only produce an eigenvalue below the Friedrichs threshold for a specific range of the interaction strength \cite[Chapter I.1]{albeverio-solvable} (and the same holds for singular perturbations of the Schr\"{o}dinger operator $-\Delta+q|x|^{-1}$ with $q>0$ -- see, e.g., \cite[Chapter I.2]{albeverio-solvable} or \cite{GM-hydrogenoid-2018}), while when $d=2$ every self-adjoint extension does have an eigenvalue below the Friedrichs threshold \cite[Chapter I.5]{albeverio-solvable}. A general characterisation of the possibility of having or not a self-adjoint extension
 with the lower bound strictly below the Friedrich's lower bound may be found in \cite{GM-2020-Friedrichs}. 
\end{remark}

Next we focus on the map $z\mapsto \hzy$ defined in \eqref{eq:HzOmegaxy} and \eqref{eq:g} above. We shall show that it determines a suitable spectral condition on $\Ev$ and displays convenient monotonicity; based on such properties we can finally deduce the strict monotonicity of $\Ev$ with respect to $\aa$.

\begin{prop}\label{prop:BS}
 Same assumptions as in Proposition \ref{prop:EVEF}. Then:
	\begin{myenum}
	\item one has $\Ev = z$ with $z \in (-\infty,\lm_1(\Omg))$ if and only if $\hzy = \aa$; 
	\item the function $(-\infty,\lm_1(\Omg))\ni z\mapsto 	\hzy$ is continuous and monotone increasing;
	\item one has $\lm_1^{\aa_1}(\Omg,y) < \lm_1^{\aa_2}(\Omg,y)$ for $\aa_1 < \aa_2$.
\end{myenum}
\end{prop}

\begin{proof}
	(i) If $\hzy = \aa$, then  the function $u = \gzy\in \ker(\sfS^* -z)$ belongs to $\dom\Op$ (Proposition~\ref{prop:OP}\,(i)). Proposition \ref{prop:EVEF}\,(v) then implies $\Ev = z$. Conversely, if $\Op u = \Ev u$ for some $u \in \dom\Op$, then  $u = c \gzy$ with $z = \Ev$ and some $c\ne 0$ 
	(Proposition~\ref{prop:EVEF}\,(v)), whence, using again the characterisation of $y\in\mathrm{dom}\Op$ from Proposition~\ref{prop:OP}\,(i), $\aa = \hzy$.

	(ii) 	Let $z_1 < z_2 < \lm_1(\Omg)$ be arbitrary. Owing to Proposition \ref{prop:EVEF}\,(iii) and\,(iv), such values are surely attained (a priori multiple times) by the function $\mathbb{R}\ni\alpha\mapsto\Ev$, and in fact it is always possible to select $\alpha_1<\alpha_2$ such that $z_1 = \lm_1^{\aa_1}(\Omg,y)$ and $z_2 = \lm_1^{\aa_2}(\Omg,y)$.
	Applying part (i) one then finds $h_{z_j,y}^\Omg = \aa_j$, $j\in\{1,2\}$. Thus, $h_{z_1,y}^\Omg < h_{z_2,y}^\Omg$.

	(iii)
	Proposition~\ref{prop:basic}\,(iii) and the min-max principle imply $\lm_1^{\aa_1}(\Omg,y) \le \lm_1^{\aa_2}(\Omg,y)$ whenever $\alpha_1<\alpha_2$. We are left with excluding the case of equality. If 
	$\lm_1^{\aa_1}(\Omg,y) = \lm_1^{\aa_2}(\Omg,y) =: z$
	held, then for any $\aa\in (\aa_1,\aa_2)$
	we would have $\Ev = z$. Hence, by~(i) we would get
	$\aa = \hzy$ for all $\aa\in(\aa_1,\aa_2)$, thus yielding an obvious contradiction.
\end{proof}

\section{The Faber-Krahn inequality for the operator $\Op$}\label{sec:FaberKrahnForHDpoint}

We are finally in the condition to formulate and prove the first main result of the present work, namely the Faber-Krahn inequality for the Dirichlet Laplacian on a bounded domain with a point interaction.

\begin{thm}\label{thm:FKDL}
Let $\Omg$ be a bounded domain in $\dR^d$ with $C^\infty$-boundary, $d\in\{2,3\}$, and let $y\in\Omg$, $\alpha\in\mathbb{R}$. Correspondingly, let $\Ev$ be the lowest eigenvalue of the operator $\Op$ qualified in Proposition \ref{prop:OP}. Let $\cB\subset\dR^d$ be a ball centred at the origin $\bfo\in\dR^d$	and satisfying $|\cB| = |\Omg|$. Then 
	\[
		\EvB \le \Ev\,.
	\]
\end{thm}

Thus: under fixed volume of the domain and fixed interaction strength parameter $\aa \in\dR$, and with generic position of point interaction centre, the principal eigenvalue is minimised by the ball with the point interaction supported at the ball's centre.

The ordinary Faber-Krahn inequality is retrieved from Theorem \ref{thm:FKDL} in the limit $\aa \arr +\infty$.

\begin{remark} We have already argued in Section \ref{sec:intro} that in a sense the result expressed by Theorem \ref{thm:FKDL} has the same spirit of the optimal placement of an `obstacle' or `impurity' within the considered domain. In fact, Theorem \ref{thm:FKDL} implies that the position $y$ of the point interaction's support which minimises the lowest eigenvalue of the Hamiltonian ${\sfH_{\alpha,y}^{\cB}}$ on the ball is precisely the ball's centre. This holds \emph{irrespective of the sign of $\alpha$}, hence of the attractive or repulsive nature of the interaction supported at $y$ (here attraction or repulsion is meant in the ordinary sense of scattering theory, thus based on the sign of the scattering length). This is not in contradiction with the fact (see the already-mentioned analysis \cite{Harrell-Kroeger-Kurata-2001}) that the optimal placement, inside the ball, of a bounded potential $V_y$ localised around $y$, in order to minimise the lowest eigenvalue of $\OpD+V_y$, is achieved by putting $y$ at the centre in case of negative potential well, and $y$ at the boundary in case of positive bump-like potential. In this respect, the point interaction modelled by ${\sfH_{\alpha,y}^{\cB}}$ has the same effect of a negative and localised potential well (whence indeed the lowering of the lowest eigenvalue of $\Op$ with respect to $\OpD$, Remark \ref{rem:bottom_always_below}). In fact, by inspection of the explicit eigenfunctions (computed in \cite[Sect.~4]{Blanchard-figari-Mantile-2007}) one sees that they concentrate around $y$, analogously to the eigenfunctions of $\OpD+V_y$ with $V_y<0$.
\end{remark}

For the proof of Theorem \ref{thm:FKDL} it is convenient to introduce the new shorthand
 \[
  \lm_\Omg \;:= \;\Ev\,,\qquad \lm_\cB \;:=\; \EvB
 \]
 at fixed $\aa\in\dR$ and $y\in \Omg$. The thesis to prove is therefore $\lm_\cB\leq \lm_\Omg$. Prior to that, we can rule out the possibility $\lm_\cB > 0$ and $\lm_\Omg \le 0$.

\begin{lemma}
 Under the assumptions of Theorem \ref{thm:FKDL} it is impossible that simultaneously $\lm_\cB > 0$ and $\lm_\Omg \le 0$.
\end{lemma}

\begin{proof}
 Assuming $\lm_\cB > 0$, Proposition~\ref{prop:basic}\,(ii) combined with Proposition~\ref{prop:BS}\,(i) imply $\aa > h^\cB_\bfo$. Moreover, $h^\cB_\bfo\geq \hy$ (Proposition~\ref{prop:Green2}\,(iii) and \eqref{eq:brevity}). Hence, $\aa > \hy$. Then Proposition \ref{prop:basic}\,(ii) and Proposition~\ref{prop:BS}\,(i) imply $\lm_\Omg > 0$.	
\end{proof}

Theorem \ref{thm:FKDL} is therefore to be proved in the only possible non-trivial scenario that $\lm_\cB$ and $\lm_\Omg$ have the same sign (either $\lm_\cB, \lm_\Omg < 0$, or $\lm_\cB, \lm_\Omg > 0$), because the case $\lm_\cB \le 0$, $\lm_\Omg \ge 0$ is trivial, and the case $\lm_\cB > 0$, $\lm_\Omg \le 0$ is impossible.

At this point, as announced in Section~\ref{sec:intro}, we find it instructive to present two alternative routes. The first applies to all cases and is entirely based on the
\emph{Bandle's inequality} $\hzy \le h_{z,\bfo}^\cB$ underlying Proposition~\ref{prop:Green2}\,(iii).  For the first proof of Theorem \ref{thm:FKDL}, we exploit Bandle's inequality in combination with our monotonicity analysis for $\Op$ developed in Section~\ref{sec:lowesteigenvalue}.

\begin{proof}[Proof of Theorem \ref{thm:FKDL} -- first version, based on the Bandle's inequality]~

Owing to Proposition~\ref{prop:BS}\,(i), $h_{\lm_\Omg,y}^\Omg =\alpha= h_{\lm_\cB,\bfo}^\cB$. Moreover, Proposition~\ref{prop:Green2}\,(iii) implies $h_{\lm_\cB,y}^\Omg \le h_{\lm_\cB,\bfo}^\cB$. Hence, $h_{\lm_\cB,y}^\Omg \le h_{\lm_\Omg,y}^\Omg$. Taking into account the increasing monotonicity of the function $(-\infty,\lm_1(\Omg))\ni z\mapsto h_{z,y}^\Omg$ (Proposition~\ref{prop:BS}\,(ii)), we conclude from $h_{\lm_\cB,y}^\Omg \le h_{\lm_\Omg,y}^\Omg$ that $\lm_\cB \le \lm_\Omg$.
\end{proof}

Next, we present an independent proof, applicable to the case $\lm_\Omg > 0$, which has a two-fold virtue. First, it provides a more direct adaptation of the original Faber-Krahn inequality's demonstration scheme to the present playground of point interaction Hamiltonian. Second, it is completely independent of Bandle's inequality and in fact reproves it by alternative means (Corollary \ref{cor:retrieving-Bandle} below).

\begin{proof}[Proof of Theorem \ref{thm:FKDL} -- second version for $\lm_\Omg > 0$, Bandle-inequality-independent]~

Let $\Ef\colon \Omg\arr\dR$ be the eigenfunction of $\Op$ corresponding to its lowest eigenvalue $\lm_\Omg$. By additional assumption, $\lm_\Omg > 0$. Then $\Ef$ decomposes as 
	\[
	 \Ef \;=\; v + \xi \gy
	\]
for some $v\in H^1_0(\Omg)$, $v \ge 0$ and $\xi \ge 0$ (Proposition~\ref{prop:EVEF}\,{\rm (vi)}). 
	Furthermore, $v^* \in H^1_0(\cB)$ (Corollary \ref{prop:rearr}\,(iii)), whence
	\[
	 v^* + \xi g_\bfo^\cB\; \in\; \dom\frh_{\aa,\bfo}^\cB
	\]
         (formula \eqref{eq:h} above). Then applying the min-max principle to the test function
	$v^* + \xi g_\bfo^\cB$ yields
	\begin{equation}\label{eq:min-max}
		\lm_\cB
		\;\le\; 
		\frac{\frmB\big[v^* + \xi g_\bfo^\cB\big]}
		{\;\| v^* + \xi g_\bfo^\cB \|^2_{L^2(\cB)}}\,.
	\end{equation}
	It suffices to show that this upper bound does not exceed $\lm_\Omg$. To this aim
	we estimate the numerator and the denominator
	separately.

	Concerning the numerator, we find
		\begin{equation}\label{eq:numerator}
	\begin{split}
		\frmB\big[v^*+\xi g_\bfo^\cB\big]  
		\;& =\;
		\|\nb v^*\|^2_{L^2(\cB)} + 
		\big(\aa-h_\bfo^\cB\big)|\xi|^2\\
		\;& \le\;
		\|\nb v\|^2_{L^2(\Omg)} + 
		\big(\aa-h_\bfo^\cB\big)|\xi|^2\\
		\;& \le\; 
		\|\nb v\|^2_{L^2(\Omg)} + 
		\big(\aa-\hy\big)|\xi|^2
		\;=\; \frm\big[v+\xi \gy\big]\,,
	\end{split}	
	\end{equation}
	having used \eqref{eq:h} in the first step, the P\'{o}lya-Szeg\H{o} inequality (Corollary~\ref{prop:rearr}\,(iii)) in the second, the bound
	\begin{equation*}
		\aa - h_{\bfo}^\cB
		\le
		\aa - h_y^\Omg.
	\end{equation*}
	in the third step (which in turn follows from Proposition~\ref{prop:Green2}\,(iii) and \eqref{eq:brevity}), and again \eqref{eq:h} in the last step.

	Concerning the denominator in \eqref{eq:min-max}, first we observe that
	\begin{equation}\label{eq:gammas_ineq}
	\ggy
	\;=\;
	\|\gy\|_{L^2(\Omg)}^{-2}
	\;=\; 
	\|(\gy)^*\|_{L^2(\cB)}^{-2}
	\;\ge \;
	\|g^\cB_\bfo\|_{L^2(\cB)}^{-2}
	\;=\; 
	\gg_\bfo^\cB\,,
	\end{equation}
	as a consequence of \eqref{eq:brevity}, of Corollary~\ref{prop:rearr}\,(i), and of  Proposition~\ref{prop:Green2}\,(ii). We then estimate
		\begin{equation}\label{eq:denominator}
	\begin{split}
		\| v^* + \xi g_\bfo^\cB \|^2_{L^2(\cB)}
		\;& =\; 
		\|v^*\|^2_{L^2(\cB)} + 
		2\xi(v^*, g_\bfo^\cB)_{L^2(\cB)}
		+(\gg_\bfo^{\cB})^{-1}|\xi|^2\\
		\;& =\;
		\|v\|^2_{L^2(\Omg)} + 
		2\xi(v^*, g_\bfo^\cB)_{L^2(\cB)}
		+(\ggy)^{-1}|\xi|^2\\
		\;& \ge\;
		\|v\|^2_{L^2(\Omg)} + 
		2\xi(v^*, (\gy)^*)_{L^2(\cB)}
		+(\ggy)^{-1}|\xi|^2\\
		\;& \ge\;
		\|v\|^2_{L^2(\Omg)} + 
		2\xi(v, \gy)_{L^2(\Omg)}
		+(\ggy)^{-1}|\xi|^2\\
		\;& =\; \|v + \xi\gy\|^2_{L^2(\Omg)}\,.
	\end{split}	
	\end{equation}
	For \eqref{eq:denominator} we used that $v,v^*\ge 0$ and $\xi \ge 0$, and we applied Corollary~\ref{prop:rearr}\,(i) and \eqref{eq:gammas_ineq} in the second step, Proposition~\ref{prop:Green2}\,(ii) and \eqref{eq:brevity} in the third, and the Hardy-Littlewood inequality (Corollary~\ref{prop:rearr}\,(ii)) in the fourth.

	Last, plugging \eqref{eq:numerator} and \eqref{eq:denominator} into \eqref{eq:min-max}, and using \eqref{eq:h}, yields
	\begin{equation*}
		\lm_\cB 
		\;\le\;
		\frac{\frm\big[v+\xi \gy\big]}{\|v+\xi \gy\|^2_{L^2(\Omg)}} \;=\; \lm_\Omg\,,
	\end{equation*}
	which completes the proof.
\end{proof}

\begin{cor}\label{cor:retrieving-Bandle}
 Let $\Omg$ be a bounded domain in $\dR^d$ with $C^\infty$-boundary, $d\in\{2,3\}$, and let $y\in\Omg$, $\alpha\in\mathbb{R}$. Let $\cB\subset\dR^d$ be a ball centred at the origin $\bfo\in\dR^d$ and satisfying $|\cB| = |\Omg|$. For $z \in (0,\lm_1(\Omg))$, the relative Green functions $h_{z,y}^\Omg$ and $h_{z,\bfo}^\cB$ defined in \eqref{eq:HzOmegaxy} and \eqref{eq:g} satisfy
 \[
  h_{z,\bfo}^\cB \;\ge\; h_{z,y}^\Omg\,.
 \]
\end{cor}

\begin{proof}
 As announced, this is rather a corollary of the second proof of Theorem \ref{thm:FKDL}.
 A generic $z \in (0,\lm_1(\Omg))$ can be written as $z=\lm_1^\aa(\Omg,y)$ for some $\alpha\in\mathbb{R}$, because the function $\alpha\mapsto\lm_1^\aa(\Omg,y)$ attains with continuity all real values below $\lambda_1(\Omega)$ (Proposition \ref{prop:EVEF}\,(iii) and (iv)). For such $\alpha$ and such positive value $\lm_1^\aa(\Omg,y)$, the second proof of Theorem \ref{thm:FKDL} shows that $\lm_1^\aa(\Omg,y)\ge\lm_1^\aa(\cB,\bfo)$. Now, on the one hand
 \[
  \alpha\;=\;h_{\lm_1^\aa(\cB,\bfo),\mathbf{0}}^{\cB}\;=\;h_{\lm_1^\aa(\Omg,y),y}^{\Omega}\;=\;h_{z,y}^\Omg
 \]
 (Proposition \ref{prop:BS}\,(i)). On the other hand, the inequality $\lm_1^\aa(\Omg,y)\ge\lm_1^\aa(\cB,\bfo)$ implies
 \[
  h_{z,\bfo}^\cB\;=\;h_{\lm_1^\aa(\Omg,y),\mathbf{0}}^{\cB}\;\geq\;h_{\lm_1^\aa(\cB,\bfo),\mathbf{0}}^{\cB}
 \]
 (Proposition \ref{prop:BS}\,(ii)). Combining these formulas yields the conclusion.
 \end{proof}

 It is admittedly remarkable that the estimate $h_{z,\bfo}^\cB \ge h_{z,y}^\Omg$ (Proposition~\ref{prop:Green2}\,(iii), Corollary \ref{cor:retrieving-Bandle}), involving Green functions of Dirichlet Laplacians, is so intimately connected with the spectral theory of the Hamiltonian with a point interaction in bounded domain.

\section{Systems with Coulomb interactions}\label{sec:Coulomb}

Let us move now on to the second main focus of the present work, namely Faber-Krahn-type inequalities
for one- and two-body systems with Coulomb interactions.

%In Subsection~\ref{ssec:1body} we treat the one-body system with an impurity. In Subsection~\ref{ssec:2body} we analyse the genuine two-body system of bosons interacting via the Coulomb potential.
%

\subsection{One-body case}\label{ssec:1body}

As outlined already in Section \ref{sec:intro}, we shall establish the following result: under fixed volume of the domain and fixed strength of the attractive Coulomb interaction between the particle and the impurity, 
the principal eigenvalue is minimised by the ball with the impurity located in its center.

For the present setting, let $\Omg\subset\dR^3$ be a bounded $C^\infty$-smooth domain and
let the point $y\in \Omg$. For $q\in\dR$, we consider, on the Hilbert space $L^2(\Omg)$, the quadratic form
\begin{equation}\label{key}
\begin{split}
 \dom\frt_{q,y}^\Omg \;&:=\; H^1_0(\Omg) \\
 \frt_{q,y}^\Omg[u] 
	\;&:=\; 
	\|\nb u\|^2_{L^2(\Omg)}
	-q\int_\Omg\frac{\;|u(x)|^2}{|x-y|}\,\dd x\,.
\end{split}
\end{equation}
In $\frt_{q,y}^\Omg$, the symmetric perturbation term 
\[
	H^1_0(\Omg)\ni u\mapsto 
	\int_\Omg\frac{\;|u(x)|^2}{|x-y|}\,\dd x
\] 
is form-bounded with respect to the kinetic energy term $H^1_0(\Omg)\ni u \mapsto 	\|\nb u\|^2_{L^2(\Omg)}$ with bound $<1$ (see, e.g., \cite[Lemma VI.4.8b]{Kato-perturbation}).  Therefore, $\frt_{q,y}^\Omg$ is closed and semi-bounded and the subspace $C^\infty_0(\Omg)$ is a core for it.
As a consequence, $\frt_{q,y}^\Omg$ is the quadratic form of a uniquely determined self-adjoint and lower semi-bounded operator on $L^2(\Omg)$, which we shall denote by $\sfT_{q,y}^\Omg$.

Now, the embedding $H^1_0(\Omg)\subset L^2(\Omg)$ is compact (see, e.g., \cite[Theorem 1.4.3.2]{Grisvard-1985}), therefore by, e.g., \cite[Proposition 10.6]{schmu_unbdd_sa} the spectrum of $\sfT_{q,y}^\Omg$ is purely discrete. We denote by $\mu_1^q(\Omg,y)$ the lowest eigenvalue  of $\sfT_{q,y}^\Omg$.

\begin{thm}\label{thm:Coulomb-1body}
Let $\Omg\subset\dR^3$ be a bounded $C^\infty$-smooth domain and let $y\in \Omg$. Let $\cB\subset\dR^3$ be the ball centred at the origin $\bfo\in\dR^3$ and satisfying $|\cB| = |\Omg|$. Let $q\geq 0$. Then 
	\[
		\mu_1^q(\cB,\bfo) \le \mu_1^q(\Omg,y)\,.
	\]
\end{thm}

\begin{remark}
 The choice $q=0$ in Theorem \ref{thm:Coulomb-1body} corresponds to the ordinary Faber-Krahn inequality.
\end{remark}

\begin{proof}[Proof of Theorem \ref{thm:Coulomb-1body}] Let us introduce the shorthand
\[
 \mu_\Omg \;:=\; \mu_1^q(\Omg,y)\,,\qquad \mu_\cB \;:= \;\mu_1^q(\cB,\bfo)\,.
\]
The thesis to prove is therefore $\mu_\cB\leq \mu_\Omg$.
We also set
\[
 V_{q,y}\colon\dR^3\arr\dR^+\,,\qquad V_{q,y}(x)\; :=\; \frac{q}{|x-y|}
\]
and observe that
\begin{equation}
 V_{q,y}^* \;=\; V_{q,\bfo}\,.
\end{equation}

	Let $u_1^q\in H^1_0(\Omg)$ be the eigenfunction
	of $\sfT_{q,y}^\Omg$ corresponding to its lowest
	eigenvalue $\mu_\Omg$. Non-restrictively, the function $u_1^q$ can be chosen to be real-valued and non-negative on $\Omg$, and moreover it satisfies		%
	\begin{equation}\label{eq:u0}
		\frt_{q,y}^\Omg[u_1^q] 
		\;=\;
		\mu_\Omg\|u_1^q\|^2_{L^2(\Omg)}\,.
	\end{equation}
	In the following we shall simply write $u:=u_1^q$.

	Let $\wt u\in H^1(\dR^3)$ be the extension of $u$ by zero.  Owing to Proposition \ref{prop:rearr0}\,(i) and (iii), also $\wt{u}^*\in H^1(\dR^3)$ with
	\begin{equation}\label{eq:u1}
	\begin{split}
	 \|\wt{u}^*\|_{L^2(\dR^3)}\;&=\;\|\wt{u}\|_{L^2(\dR^3)} \\
	 \|\nb \wt{u}^*\|_{L^2(\dR^3)} \;&\leq\;\|\nb \wt{u}\|_{L^2(\dR^3)}\,.
	\end{split}
	\end{equation}
	Moreover, owing to Lemma \ref{lem:rearr}\,(iv) and (v),
	\begin{equation}\label{eq:utildemore}
	 \begin{split}
	  \supp\wt{u}^* \;&\subset\;\cB \\
	  (\wt{u}^*)^2 \;&=\; (\wt{u}^2)^*\,.
	 \end{split}
	\end{equation}
       From $\wt{u}^* \in H^1(\dR^3)$ and $\supp\wt{u}^* \subset \cB$, we infer that
       \begin{equation}\label{eq:vtrunc}
        v \;:=\; \wt{u}^*|_\cB\;\in\; H^1_0(\cB)\,,
       \end{equation}
       and using \eqref{eq:utildemore} and the Hardy-Littlewood inequality (Proposition~\ref{prop:rearr0}\,(ii)) we find
	\begin{equation}\label{eq:u2}
	%\begin{aligned}
		\int_{\dR^3}
		V_{q,\bfo}(x)(\wt{u}^*(x))^2\dd x
		 \;=\;
		\int_{\dR^3}
		V_{q,\bfo}(x)(\wt{u}^2)^*(x)\dd x
		\;\ge\;
		\int_{\dR^3}
		V_{q,y}(x) \wt{u}^2(x)\dd x\,.
	%\end{aligned}
	\end{equation}

	Now, using $v$ as a trial function for the Rayleigh quotient for the operator $\sfT_{q,\bfo}^\cB$, we finally find
		\[
	\begin{aligned}
		\mu_\cB  
		\;\le\; \frac{\frt_{q,\bfo}^\cB[v]}{\;\|v\|^2_{L^2(\cB)}}
	 	\;& =\;
	 	 \frac{\|\nb \wt{u}^*\|^2_{L^2(\dR^3)}
	 		- {\displaystyle\int_{\dR^3}}V_{q,\bfo}(x)
	 		|\wt{u}^*(x)|^2\dd x}{\|\wt{u}\|^2_{L^2(\dR^3)}}\\
		 \;& \le\;
		 \frac{\|\nb \wt{u}\|^2_{L^2(\dR^3)}
		 	- \displaystyle{\int_{\dR^3}}
		 	V_{q,y}|\wt{u}(x)|^2\dd x}{\|\wt{u}\|^2_{L^2(\dR^3)}}
		 \; =\;
		 \frac{\frt_{q,y}^\Omg[u]}{\;\|u\|^2_{L^2(\Omg)}}
		 \;=\;
		 \mu_\Omg\,,
	\end{aligned}
	\]
	where we used the min-max principle in the first step, \eqref{eq:u1} in the second, \eqref{eq:vtrunc}, \eqref{eq:u2} in the third, \eqref{key} in the fourth, and \eqref{eq:u0} in the last.
\end{proof}

\subsection{Two-body case}\label{ssec:2body}
Next, we examine the model for a two-body quantum system on a bounded domain with inter-particle attractive Coulomb interaction.

% of strength $q > 0$ is not smaller than the principal eigenvalue of the two-particle system on 
% the ball of the same volume, in which the particles
% interact only with the impurity located in the center
% of the ball via the attractive Coulomb potential
% of strength $q / 2$.

As before, let $\Omg\subset\dR^3$ be a $C^\infty$-smooth domain. For $q\in\dR$, we consider, on the Hilbert space $L^2(\Omega\times\Omega)$, the quadratic form
\begin{equation}\label{key2}
\begin{split}
 \dom\frt_q^\Omg \;&:=\; H^1_0(\Omg\tm \Omg) \\
 \frt_q^\Omg[u] 
	\;&:=\;
	\|\nb u\|^2_{L^2(\Omg\tm \Omg)}
	-q\int_{\Omg\tm\Omg}
	\frac{\,|u(x_1,x_2)|^2}{|x_1-x_2|}\,\dd x_1\dd x_2\,.
\end{split}
\end{equation}
In $\frt_q^\Omg$, the symmetric perturbation term 
\[
	H^1_0(\Omg\tm\Omg)\ni u\mapsto 
	\int_{\Omg\tm\Omg}
	\frac{\;|u(x_1,x_2)|^2}{|x_1-x_2|}
	\,\dd x_1 \dd x_2
\]
is form-bounded with respect to the kinetic energy term 
\[
	H^1_0(\Omg\tm\Omg)\ni u \mapsto 	\|\nb u\|^2_{L^2(\Omg\tm\Omg)}
\]
with the bound $<1$ (see, e.g., \cite[Lemma 4]{Kato-1951}). Therefore, $\frt_q^\Omg$ is closed and semi-bounded and the subspace $C^\infty_0(\Omg\times\Omega)$ is a core for it.
As a consequence, $\frt_q^\Omg$ is the quadratic form of a uniquely determined self-adjoint and lower semi-bounded operator on $L^2(\Omg)$, which we shall denote by $\sfT_q^\Omg$. As the embedding $H^1_0(\Omg\tm\Omg)\subset L^2(\Omg\tm\Omg)$ is compact, the spectrum of $\sfT_q^\Omg$ is purely discrete. We denote by $\nu_1^q(\Omg)$ the lowest eigenvalue.

\begin{thm}\label{thm:Coulomb2}
Let $\Omg\subset\dR^3$ be a $C^\infty$-smooth domain and let $\cB\subset\dR^3$ be a ball centred at the origin $\bfo\in\dR^3$ and satisfying $|\cB| = |\Omg|$. Let $q\geq 0$. Then
	\[
		2\mu_1^{q/2}(\cB,\bfo)\; \le\; \nu_1^q(\Omg)\,.
	\]	
\end{thm}

\begin{remark}
 With the choice $q=0$ in Theorem \ref{thm:Coulomb2}, namely when the two particles are uncoupled, the quantity $\nu_1^0(\Omg)$ is obviously twice as the ground state energy of a single particle confined in $\Omega$, namely, with Hamiltonian given by the Dirichlet Laplacian on $\Omega$. In this case Theorem \ref{thm:Coulomb2} implies
 \[
  2\lambda_1(\Omega)\;=\;\nu_1^0(\Omg)\;\geq\;2\mu_1^{0}(\cB,\bfo)\;=\;2\lambda_1(\cB)\,,
 \]
and one retrieves the ordinary Faber-Krahn inequality. 
\end{remark}

\begin{proof} Let us introduce the shorthand
\[
 \nu_\Omg \;:= \;\nu_1^q(\Omg)\,,\qquad \mu_\cB \;:= \;\mu_1^{q/2}(\cB,\bfo)\,.
\]
The thesis to prove is therefore $2\mu_\cB\leq \mu_\Omg$.
We also set 
\[
 V_{q,x_1}\colon\dR^3\arr\dR^+\,,\qquad V_{q,x_1}(x_2) \;:=\; \frac{q}{|x_1-x_2|}
\]
and observe that 
\begin{equation}\label{eq:Vrearranged}
 V_{q,x_1}^* \;=\; V_{q,\bfo}
\end{equation}
independently of $x_1$.
	
	Let  $u_1^q \in H^1_0(\Omg\tm\Omg)$ be the eigenfunction of $\sfT_q^\Omg$ corresponding to its lowest eigenvalue $\nu_\Omg$. Non-restrictively the function $u_1^q$ can be chosen to be real-valued and non-negative, and moreover it satisfies
	\begin{equation}\label{eq:u02}
		\frt_q^\Omg[u_1^q] \;=\;\nu_\Omg\|u_1^q\|^2_{L^2(\Omg\tm \Omg)}\,.
	\end{equation}
	It can also be easily shown that
	\begin{equation}\label{eq:symmetry}
		u_1^q(x_1,x_2)\;=\; u_1^q(x_2,x_1)\qquad\quad  (x_1,x_2\in\Omg)
	\end{equation}
	(bosonic and absolute ground state coincide).
	In the following we shall simply write $u:=u_1^q$.

	Let $\wt u\in H^1(\dR^3\times\mathbb{R}^3)$ be the extension of $u$ by zero, and let us consider the Steiner rearrangements $S_1\wt{u}$ and $S_2\wt{u}$ of $\wt u$ respectively with respect to the first and the second three-dimensional variable (as defined in \eqref{eq:Steiner}), as well as, correspondingly, the further rearrangements $S_1S_2\wt{u}$ and $S_2S_1\wt{u}$ with respect to the other variable. By construction,
	\begin{equation}\label{eq:supports}
         \begin{split}
          \supp (S_1\wt{u}) \;&\subset\; \cB\tm\Omg \\
          \supp (S_2\wt{u}) \;&\subset\; \Omg\tm\cB \\
          \supp(S_1S_2\wt{u}) \;&\subset\; \cB\tm\cB\,.
         \end{split}
        \end{equation}
        Moreover, as a consequence of the exchange symmetry \eqref{eq:symmetry} (Proposition \ref{prop:rearr2}\,(iii)),
	\begin{equation}\label{eq:bosonicsym}
	 S_1S_2\wt{u}\;=\; S_2S_1\wt{u}\,.
	\end{equation}

	Owing to Proposition \ref{prop:rearr2}, $S_1\wt{u}, S_2\wt{u},S_1S_2\wt{u}, S_2S_1\wt{u} \in H^1(\mathbb{R}^3\times\mathbb{R}^3)$ with 
	\begin{equation}\label{eq:u3}
	 \begin{split}
	 \|S_1S_2\wt{u}\|_{L^2(\mathbb{R}^3\times\mathbb{R}^3)}\;&=\;\|\wt{u}\|_{L^2(\mathbb{R}^3\times\mathbb{R}^3)} \\
	 \|\nb (S_1S_2\wt{u})\|_{L^2(\mathbb{R}^3\times\mathbb{R}^3)}\;&\leq\;\|\nb \wt{u}\|_{L^2(\mathbb{R}^3\times\mathbb{R}^3)}\,.
	 \end{split}
	\end{equation}

	Next, we observe that for (almost every) fixed $x_1\in \cB$ (and trivially for $x_1\in\mathbb{R}^3\setminus\cB$), the function $(S_1\widetilde{u})(x_1,\cdot)$ is square integrable on the whole $\mathbb{R}^3$ and its symmetric decreasing rearrangement is nothing but $(S_2S_1\widetilde{u})(x_1,\cdot)$. Thus, Proposition \ref{prop:rearr0}\,(i) gives
	\begin{equation}\label{eq:slice1}
	\begin{split}
	 \int_{\mathbb{R}^3}\dd x_2\,\big|(S_2S_1\widetilde{u})(x_1,x_2)\big|^2\;&=\;\int_{\mathbb{R}^3}\dd x_2\,\big|\big((S_1\widetilde{u})(x_1,\cdot)\big)^*(x_2)\big|^2 \\
	 &=\; \int_{\mathbb{R}^3}\dd x_2\,\big|(S_1\widetilde{u})(x_1,x_2)\big|^2\,,
	\end{split}
	\end{equation}
        and analogously,
        \begin{equation}\label{eq:slice2}
         \int_{\mathbb{R}^3}\dd x_1\,\big|(S_1S_2\widetilde{u})(x_1,x_2)\big|^2\;=\; \int_{\mathbb{R}^3}\dd x_2\,\big|(S_2\widetilde{u})(x_1,x_2)\big|^2\,.
        \end{equation}
       In turn, \eqref{eq:slice1}-\eqref{eq:slice2} (together with \eqref{eq:bosonicsym}) imply
	\begin{equation*}
	\begin{aligned}
		&\iint_{\mathbb{R}^3\times\mathbb{R}^3} |(S_1S_2\wt{u})(\bx)|^2
		\big[V_{\frac{q}{2},\bfo}(x_1) + V_{\frac{q}{2},\bfo}(x_2)\big]\dd x_1 \dd x_2 \\
		& \quad =
		\iint_{\mathbb{R}^3\times\mathbb{R}^3} |(S_1\wt{u})(\bx)|^2
		V_{\frac{q}{2},\bfo}(x_1)\dd x_1 \dd x_2
		+
		\iint_{\mathbb{R}^3\times\mathbb{R}^3} |(S_2\wt{u})(\bx)|^2
		V_{\frac{q}{2},\bfo}(x_2)\dd x_1 \dd x_2\,,
	\end{aligned}
	\end{equation*}
	having used the shorthand $\bx = (x_1,x_2)$. For each summand of the r.h.s.~above, the Hardy-Littlewood inequality (Proposition~\ref{prop:rearr0}\,(ii)) and \eqref{eq:Vrearranged} yield
	\[
	 \begin{split}
	  \int_{\dR^3}|(S_1\wt{u})(\bx)|^2
		V_{\frac{q}{2},\bfo}(x_1)\dd x_1\;&\geq\;\int_{\dR^3}|\wt{u}(\bx)|^2
		V_{\frac{q}{2},x_2}(x_1)\dd x_1 \\
	\int_{\dR^3}|(S_2\wt{u})(\bx)|^2
		V_{\frac{q}{2},\bfo}(x_2)\dd x_2\;&\geq\;\int_{\dR^3}|\wt{u}(\bx)|^2
		V_{\frac{q}{2},x_1}(x_2)\dd x_2\,,	
	 \end{split}
	\]
	and obviously $V_{\frac{q}{2},x_2}(x_1)+V_{\frac{q}{2},x_1}(x_2)=V_{q,x_1}(x_2)$, whence finally
	\begin{equation}\label{eq:u5}
	\begin{aligned} 
		& \iint_{\mathbb{R}^3\times\mathbb{R}^3} |(S_1S_2\wt{u})(\bx)|^2
		\big[V_{\frac{q}{2},\bfo}(x_1) + V_{\frac{q}{2},\bfo}(x_2)\big]\dd x_1 \dd x_2\\
		& \qquad \geq 
		\int_{\mathbb{R}^3\times\mathbb{R}^3}
		|\wt{u}(\bx)|^2
		V_{q,x_1}(x_2)\,\dd \bx\,.
	\end{aligned}
	\end{equation}

	Now, by construction (see \eqref{eq:supports} and \eqref{eq:u3} above),
	\begin{equation}\label{eq:newtrialv}
	 v \;:=\; (S_1S_2\wt{u})|_{\cB\tm\cB}\;\in\; H^1_0(\cB\tm\cB)\,.
	\end{equation}
        Using $v$ as a trial function for the Rayleigh quotient for the	self-adjoint operator 
	\[
		\sfT_{\frac{q}{2},\bfo}^\cB\otimes \mathbbm{1} + \mathbbm{1} \otimes\sfT_{\frac{q}{2},\bfo}^\cB
	\]
	acting in the Hilbert space $L^2(\cB\tm\cB)$ and we finally find
	\[
	\begin{aligned}
		2\mu_\cB \;& 
		\le\; \frac{\;\|\nb v\|^2_{L^2(\cB\tm\cB)}
			- {\displaystyle\int_{\cB\tm \cB}}
			|v(\bx)|^2\left[ V_{\frac{q}{2},\bfo}(x_1) + V_{\frac{q}{2},\bfo}(x_2)\right]\dd \bx\;}{\|v\|^2_{L^2(\cB\tm \cB)}}\\
		&=\; \frac{\;\|\nb (S_1S_2\wt{u})\|^2_{L^2(\mathbb{R}^3\times\mathbb{R}^3)}
		- {\displaystyle\int_{\mathbb{R}^3\times\mathbb{R}^3}}
		|(S_1S_2\wt{u})(\bx)|^2\left[ V_{\frac{q}{2},\bfo}(x_1) + V_{\frac{q}{2},\bfo}(x_2)\right]\dd \bx\;}{\|(S_1S_2\wt{u})\|^2_{L^2(\mathbb{R}^3\times\mathbb{R}^3)}}\\
		& \le\;
		\frac{\;\|\nb \wt{u}\|^2_{L^2(\mathbb{R}^3\times\mathbb{R}^3)}
		- \displaystyle{\int_{\mathbb{R}^3\times\mathbb{R}^3}}
		V_{q,x_1}(x_2)|\wt{u}(\bx)|^2\dd \bx\;}{\|\wt{u}\|^2_{L^2(\mathbb{R}^3\times\mathbb{R}^3)}} \\
		&=\;\frac{\frt_q^\Omg[u]}{\;\|u\|^2_{L^2(\Omg\tm\Omg)}}\;=\; \nu_\Omg\,,
	\end{aligned}
	\]
	where we used (twice) the min-max principle in the first step, \eqref{eq:newtrialv} in the second, \eqref{eq:u3} and \eqref{eq:u5} in the third, \eqref{key2} in the fourth, and \eqref{eq:u02} in the last.
\end{proof}

\subsection*{Acknowledgement}
We are most grateful to the International Centre for Mathematical Research FBK/CIRM Trento, 
under the auspices of which a large part of this project was carried on within a research in pairs programme. This work is also partially supported by the Alexander von Humboldt foundation. 

%---------------------------------------------------

% 
% \bibliographystyle{siam}
% \bibliography{bib_ALE}

\begin{thebibliography}{10}

\bibitem{albeverio-solvable}
{\sc S.~Albeverio, F.~Gesztesy, R.~H{\o}egh-Krohn, and H.~Holden}, {\emph{Solvable {M}odels in {Q}uantum {M}echanics}}, {Texts and Monographs in
  Physics}, Springer-Verlag, New York, 1988.

\bibitem{AHK-1981-JOPTH}
{\sc S.~Albeverio and R.~H{\o}egh-Krohn}, {\emph{Point interactions as limits of
  short range interactions}}, J. Operator Theory, 6 (1981), pp.~313--339.

\bibitem{ABLOB-2020}
{\sc P.~R.~S. Antunes, R.~Benguria, V.~Lotoreichik, and
  T.~Ourmi{\`e}res-Bonafos}, \emph{A variational formulation for Dirac
  operators in bounded domains. Applications to spectral geometric
  inequalities}, arXiv.org:2003.04061 (2020).

%\bibitem{Artamoshin-2016}
%{\sc S.~Artamoshin}, \emph{Lower bounds for the first {D}irichlet eigenvalue of
%  the {L}aplacian for domains in hyperbolic space}, Math. Proc. Cambridge
%  Philos. Soc., 160 (2016), pp.~191--208.

\bibitem{Baernstein-symm-2019}
{\sc I.~A. Baernstein}, {\emph {Symmetrization in analysis}}, vol.~36 of {New
  Mathematical Monographs}, Cambridge University Press, Cambridge, 2019.
\newblock With David Drasin and Richard S. Laugesen, With a foreword by Walter
  Hayman.

\bibitem{Bandle-78}
{\sc C.~Bandle}, {\em{Estimates for the {G}reen's functions of elliptic
  operators}}, SIAM J. Math. Anal., 9 (1978), pp.~1126--1136.

\bibitem{Bandle-book-1980}
\leavevmode\vrule height 2pt depth -1.6pt width 23pt, {\emph{Isoperimetric
  inequalities and applications}}, vol.~7 of {Monographs and Studies in
  Mathematics}, Pitman (Advanced Publishing Program), Boston, Mass.-London,
  1980.

\bibitem{Bandle-Flucher-1996}
{\sc C.~Bandle and M.~Flucher}, {\emph{Harmonic radius and concentration of
  energy; hyperbolic radius and {L}iouville's equations {$\Delta U=e^U$} and
  {$\Delta U=U^{(n+2)/(n-2)}$}}}, SIAM Rev., 38 (1996), pp.~191--238.

\bibitem{Benguria-Fournais-Stock-vDB-2017}
{\sc R.~D. Benguria, S.~r. Fournais, E.~Stockmeyer, and H.~{Van Den Bosch}},
  {\emph{Spectral gaps of {D}irac operators describing graphene quantum dots}},
  Math. Phys. Anal. Geom., 20 (2017), pp.~Paper No. 11, 12.

\bibitem{Benguria-Linde-Loewe-2011}
{\sc R.~D. Benguria, H.~Linde, and B.~Loewe}, {\emph{Isoperimetric inequalities
  for eigenvalues of the {L}aplacian and the {S}chr{\"o}dinger operator}},
  Bull. Math. Sci., 2 (2012), pp.~1--56.

\bibitem{Bethe_Peierls-1935}
{\sc H.~Bethe and R.~Peierls}, {\emph{Quantum Theory of the Diplon}},
  Proceedings of the Royal Society of London. Series A, Mathematical and
  Physical Sciences, 148 (1935), pp.~146--156.

\bibitem{Bhattacharya-1999}
{\sc T.~Bhattacharya}, {\emph{A proof of the {F}aber-{K}rahn inequality for the
  first eigenvalue of the {$p$}-{L}aplacian}}, Ann. Mat. Pura Appl. (4), 177
  (1999), pp.~225--240.

\bibitem{Birman-Solomjak-1980book}
{\sc M.~S. Birman and M.~Z. Solomjak}, {\emph{Spectral theory of selfadjoint
  operators in {H}ilbert space}}, {Mathematics and its Applications (Soviet
  Series)}, D. Reidel Publishing Co., Dordrecht, 1987.
\newblock Translated from the 1980 Russian original by S. Khrushch{\"e}v and V.
  Peller.

\bibitem{Blanchard-figari-Mantile-2007}
{\sc P.~Blanchard, R.~Figari, and A.~Mantile}, {\emph{Point interaction
  {H}amiltonians in bounded domains}}, J. Math. Phys., 48 (2007), pp.~082108,
  18.

%\bibitem{Borell-1975}
%{\sc C.~Borell}, {\emph{The {B}runn-{M}inkowski inequality in {G}auss space}},
%  Invent. Math., 30 (1975), pp.~207--216.

\bibitem{Brascamp-Lieb-Luttinger-1974}
{\sc H.~J. Brascamp, E.~H. Lieb, and J.~M. Luttinger}, {\emph{A general
  rearrangement inequality for multiple integrals}}, J. Functional Analysis, 17
  (1974), pp.~227--237.

\bibitem{Brock-2001}
{\sc F.~Brock}, {\emph{An isoperimetric inequality for eigenvalues of the
  {S}tekloff problem}}, ZAMM Z. Angew. Math. Mech., 81 (2001), pp.~69--71.

\bibitem{Brock-Solynin-2000}
{\sc F.~Brock and A.~Y. Solynin}, {\emph{An approach to symmetrization via
  polarization}}, Trans. Amer. Math. Soc., 352 (2000), pp.~1759--1796.

\bibitem{Bucur-Daners-2010}
{\sc D.~Bucur and D.~Daners}, {\emph{An alternative approach to the
  {F}aber-{K}rahn inequality for {R}obin problems}}, Calc. Var. Partial
  Differential Equations, 37 (2010), pp.~75--86.

\bibitem{BFNT-2018}
{\sc D.~Bucur, V.~Ferone, C.~Nitsch, and C.~Trombetti}, {\emph{The quantitative
  {F}aber-{K}rahn inequality for the {R}obin {L}aplacian}}, J. Differential
  Equations, 264 (2018), pp.~4488--4503.

\bibitem{Bucur-Freitas-Kennedy-2017Robin}
{\sc D.~Bucur, P.~Freitas, and J.~Kennedy}, {\emph{The {R}obin problem}}, in
  {Shape optimization and spectral theory}, De Gruyter Open, Warsaw, 2017,
  pp.~78--119.

\bibitem{Bucur-Giacomini-2019}
{\sc D.~Bucur and A.~Giacomini}, {\emph{Minimization of the {$k$}-th eigenvalue
  of the {R}obin-{L}aplacian}}, J. Funct. Anal., 277 (2019), pp.~643--687.

\bibitem{Capriani-2014-Steiner}
{\sc G.~M. Capriani}, {\emph{The {S}teiner rearrangement in any codimension}},
  Calc. Var. Partial Differential Equations, 49 (2014), pp.~517--548.

\bibitem{CGIKO-2000}
{\sc S.~Chanillo, D.~Grieser, M.~Imai, K.~Kurata, and I.~Ohnishi}, {\emph{Symmetry breaking and other phenomena in the optimization of eigenvalues for
  composite membranes}}, Comm. Math. Phys., 214 (2000), pp.~315--337.

\bibitem{Chavel-1984}
{\sc I.~Chavel}, {\emph{Eigenvalues in {R}iemannian geometry}}, vol.~115 of
  {Pure and Applied Mathematics}, Academic Press, Inc., Orlando, FL, 1984.
\newblock Including a chapter by Burton Randol, With an appendix by Jozef
  Dodziuk.


\bibitem{CdV82}
{\sc Y.~Colin de Verdi\`{e}re}, \emph{Pseudo-laplaciens. I.}, Ann. Inst. Fourier, 32 (1982), pp.~275-286.


\bibitem{Cupini-Vecchi-2019}
{\sc G.~Cupini and E.~Vecchi}, {\emph{Faber-{K}rahn and {L}ieb-type inequalities
  for the composite membrane problem}}, Commun. Pure Appl. Anal., 18 (2019),
  pp.~2679--2691.

\bibitem{Dai-Fu-2011}
{\sc Q.-y. Dai and Y.-x. Fu}, {\emph{Faber-{K}rahn inequality for {R}obin
  problems involving {$p$}-{L}aplacian}}, Acta Math. Appl. Sin. Engl. Ser., 27
  (2011), pp.~13--28.

\bibitem{Daners-2006}
{\sc D.~Daners}, {\emph{A {F}aber-{K}rahn inequality for {R}obin problems in any
  space dimension}}, Math. Ann., 335 (2006), pp.~767--785.

\bibitem{Daners-Kennedy-2007}
{\sc D.~Daners and J.~Kennedy}, {\emph{Uniqueness in the {F}aber-{K}rahn
  inequality for {R}obin problems}}, SIAM J. Math. Anal., 39 (2007/08),
  pp.~1191--1207.


\bibitem{Exner-Lotoreichik-2020}
{\sc P.~Exner and V.~Lotoreichik}, {\emph{Spectral optimization for Robin
  Laplacian in domains without cut loci}}, arXiv.org:2001.02718 (2020).

\bibitem{Exner-Mantile-2007}
{\sc P.~Exner and A.~Mantile}, {\emph{On the optimization of the principal
  eigenvalue for single-centre point-interaction operators in a bounded
  region}}, J. Phys. A, 41 (2008), pp.~065305, 15.

\bibitem{Faber-1923}
{\sc G.~Faber}, {\emph{Beweis, dass unter allen homogenen Membranen von gleicher
  Fl{\"a}che und gleicher Spannung die kreisf{\"o}rmige den tiefsten Grundton
  gibt}}, in {Sitzungberichte der mathematisch-physikalischen Klasse der
  Bayerischen Akademie der Wissenschaften zu M{\"u}nchen Jahrgang}, 1923,
  pp.~169--172.

\bibitem{Fournais-Helffer-2019}
{\sc S.~Fournais and B.~Helffer}, {\emph{Inequalities for the lowest magnetic
  {N}eumann eigenvalue}}, Lett. Math. Phys., 109 (2019), pp.~1683--1700.

\bibitem{Freitas-2015}
{\sc P.~Freitas and D.~Krej\v{c}i\v{r}{\'i}k}, {\emph{The first {R}obin
  eigenvalue with negative boundary parameter}}, Adv. Math., 280 (2015),
  pp.~322--339.

\bibitem{GM-hydrogenoid-2018}
{\sc M.~Gallone and A.~Michelangeli}, {\emph{Hydrogenoid spectra withcentral
  perturbations}}, Rep. Math. Phys., 84 (2019), pp.~215--243.

\bibitem{GM-2020-Friedrichs}
\leavevmode\vrule height 2pt depth -1.6pt width 23pt, {\emph{Self-adjoint
  extensions with Friedrichs lower bound}}, arXiv.org:2003.09631 (2020).

\bibitem{GMO-KVB2017}
{\sc M.~Gallone, A.~Michelangeli, and A.~Ottolini}, {\emph  {Kre{\u\i}n-Vi\v{s}ik-Birman self-adjoint extension theory revisited}}, in
  {Mathematical Challenges of Zero Range Physics}, A.~Michelangeli, ed.,
  {INdAM-Springer series}, Springer International Publishing, 2020, SISSA
  preprint 25/2017/MATE.

\bibitem{Grisvard-1985}
{\sc P.~Grisvard}, {\emph{Elliptic problems in nonsmooth domains}}, vol.~69 of
  {Classics in Applied Mathematics}, Society for Industrial and Applied
  Mathematics (SIAM), Philadelphia, PA, 2011.
\newblock Reprint of the 1985 original [ MR0775683], With a foreword by Susanne
  C. Brenner.

\bibitem{Harrell-Kroeger-Kurata-2001}
{\sc E.~M. Harrell, P.~Kr{\"o}ger, and K.~Kurata}, {\emph{On the placement of
  an obstacle or a well so as to optimize the fundamental eigenvalue}}, SIAM J.
  Math. Anal., 33 (2001), pp.~240--259.

\bibitem{Henrot-Extremum-Problems-2006}
{\sc A.~Henrot}, {\emph{Extremum problems for eigenvalues of elliptic
  operators}}, {Frontiers in Mathematics}, Birkh{\"a}user Verlag, Basel, 2006.

\bibitem{Hu-Dai-2014}
{\sc H.~Hu and Q.~Dai}, {\emph{Isoperimetric inequalities for positive solution
  of {P}-{L}aplacian}}, Math. Inequal. Appl., 17 (2014), pp.~1453--1469.

\bibitem{Kato-1951}
{\sc T.~Kato}, {\emph{Fundamental properties of {H}amiltonian operators of
  {S}chr{\"o}dinger type}}, Trans. Amer. Math. Soc., 70 (1951), pp.~195--211.

\bibitem{Kato-perturbation}
\leavevmode\vrule height 2pt depth -1.6pt width 23pt, {\emph{Perturbation theory
  for linear operators}}, {Classics in Mathematics}, Springer-Verlag, Berlin,
  1995.
\newblock Reprint of the 1980 edition.

\bibitem{Keady-W-2018}
{\sc G.~Keady and B.~Wiwatanapataphee}, {\emph{Inequalities for the fundamental
  {R}obin eigenvalue for the {L}aplacian on {$N$}-dimensional rectangular
  parallelepipeds}}, Math. Inequal. Appl., 21 (2018), pp.~911--930.

\bibitem{Kennedy-2008}
{\sc J.~Kennedy}, {\emph{A {F}aber-{K}rahn inequality for the {L}aplacian with
  generalised {W}entzell boundary conditions}}, J. Evol. Equ., 8 (2008),
  pp.~557--582.

\bibitem{Kennedy-2010}
{\sc J.~B. Kennedy}, {\emph{On the isoperimetric problem for the higher
  eigenvalues of the {R}obin and {W}entzell {L}aplacians}}, Z. Angew. Math.
  Phys., 61 (2010), pp.~781--792.

\bibitem{Kesavan-symm-2006}
{\sc S.~Kesavan}, {\emph{Symmetrization \& applications}}, vol.~3 of {Series in
  Analysis}, World Scientific Publishing Co. Pte. Ltd., Hackensack, NJ, 2006.

\bibitem{Khalile-Lotoreichik-2019}
{\sc M.~Khalile and V.~Lotoreichik}, {\emph{Spectral isoperimetric inequalities
  for Robin Laplacians on 2-manifolds and unbounded cones}},
  arXiv.org:1909.10842 (2019).

\bibitem{Kovaric-miedbc-2014}
{\sc H.~Kova\v{r}{\'i}k}, {\emph{On the lowest eigenvalue of {L}aplace operators
  with mixed boundary conditions}}, J. Geom. Anal., 24 (2014), pp.~1509--1525.

\bibitem{Krahn-1925}
{\sc E.~Krahn}, {\emph{{\"U}ber eine von {R}ayleigh formulierte
  {M}inimaleigenschaft des {K}reises}}, Math. Ann., 94 (1925), pp.~97--100.

\bibitem{Kreicik-Lotoreichik-2018-extCompactSets}
{\sc D.~Krej\v{c}i\v{r}{\'i}k and V.~Lotoreichik}, {\emph{Optimisation of the
  lowest {R}obin eigenvalue in the exterior of a compact set}}, J. Convex
  Anal., 25 (2018), pp.~319--337.

\bibitem{Kreicik-Lotoreichik-2020-extCompactSets-II}
\leavevmode\vrule height 2pt depth -1.6pt width 23pt, {\emph{Optimisation of the
  lowest {R}obin eigenvalue in the exterior of a compact set II: non-convex
  domains and higher dimensions}}, Potential Analysis,  (2020).

\bibitem{Laugesen-2019Robin}
{\sc R.~S. Laugesen}, {\emph{The {R}obin {L}aplacian---{S}pectral conjectures,
  rectangular theorems}}, J. Math. Phys., 60 (2019), pp.~121507, 31.

\bibitem{Leoni-2017-Sobolev}
{\sc G.~Leoni}, {\emph{A first course in {S}obolev spaces}}, vol.~181 of
  {Graduate Studies in Mathematics}, American Mathematical Society, Providence,
  RI, second~ed., 2017.

\bibitem{Li-Yau-1980}
{\sc P.~Li and S.~T. Yau}, {\emph{Estimates of eigenvalues of a compact
  {R}iemannian manifold}}, in {Geometry of the {L}aplace operator ({P}roc.
  {S}ympos. {P}ure {M}ath., {U}niv. {H}awaii, {H}onolulu, {H}awaii, 1979)},
  {Proc. Sympos. Pure Math., XXXVI}, Amer. Math. Soc., Providence, R.I., 1980,
  pp.~205--239.

\bibitem{Lieb-Loss-Analysis}
{\sc E.~H. Lieb and M.~Loss}, {\emph{Analysis}}, vol.~14 of {Graduate Studies in
  Mathematics}, American Mathematical Society, Providence, RI, second~ed.,
  2001.

\bibitem{MO-2016}
{\sc A.~Michelangeli and A.~Ottolini}, {\emph{On point interactions realised as
  {T}er-{M}artirosyan-{S}kornyakov {H}amiltonians}}, Rep. Math. Phys., 79
  (2017), pp.~215--260.

\bibitem{Posilicano-2013}
{\sc A.~Posilicano}, {\emph{On the many {D}irichlet {L}aplacians on a non-convex
  polygon and their approximations by point interactions}}, J. Funct. Anal.,
  265 (2013), pp.~303--323.

\bibitem{Rayleigh-sound-1894}
{\sc J. W.~S. Rayleigh}, {\emph{The {T}heory of {S}ound}}, Dover
  Publications, New York, N. Y., 1945.
\newblock 2d ed (republication of the 1894/1896 edition).

\bibitem{rs1}
{\sc M.~Reed and B.~Simon}, {\emph{Methods of {M}odern {M}athematical
  {P}hysics}}, vol.~1, New York Academic Press, 1972.

\bibitem{Savo-2020}
{\sc A.~Savo}, {\emph{Optimal eigenvalue estimates for the {R}obin {L}aplacian
  on {R}iemannian manifolds}}, J. Differential Equations, 268 (2020),
  pp.~2280--2308.

\bibitem{schmu_unbdd_sa}
{\sc K.~Schm{\"u}dgen}, {\emph{Unbounded self-adjoint operators on {H}ilbert
  space}}, vol.~265 of {Graduate Texts in Mathematics}, Springer, Dordrecht,
  2012.

\bibitem{Szego-1954}
{\sc G.~Szeg{\"o}}, {\emph{Inequalities for certain eigenvalues of a membrane of
  given area}}, J. Rational Mech. Anal., 3 (1954), pp.~343--356.

\bibitem{Weinberger-1956}
{\sc H.~F. Weinberger}, {\emph{An isoperimetric inequality for the
  {$N$}-dimensional free membrane problem}}, J. Rational Mech. Anal., 5 (1956),
  pp.~633--636.

\bibitem{Weinstock-1954}
{\sc R.~Weinstock}, {\emph{Inequalities for a classical eigenvalue problem}}, J.
  Rational Mech. Anal., 3 (1954), pp.~745--753.

\bibitem{Xu-1995}
{\sc Y.~Xu}, {\emph{The first nonzero eigenvalue of {N}eumann problem on
  {R}iemannian manifolds}}, J. Geom. Anal., 5 (1995), pp.~151--165.

\end{thebibliography}

\def\cprime{$'$}

\bibliographystyle{alpha}

\end{document}